\documentclass[utf8,11pt,letterpaper]{article}
\input{preamble.tex}

\newcommand{\rad}[1]{{\sf rad}{#1}}
\newcommand{\OPT}{{\sf OPT}}
\usepackage[utf8]{inputenc}
\usepackage{thm-restate}
\usepackage{authblk}
\usepackage{amsthm,amsmath,amssymb}
\usepackage{hyperref}
\usepackage{thm-restate}
\usepackage{thmtools}
\usepackage{tcolorbox}
\newtheorem{rmk}[theorem]{Remark}

\theoremstyle{definition}

\newcommand{\radii}{\textsc{Capacitated Sum of Radii}\xspace}
\newcommand{\sumradii}{\textsc{Sum of Radii}\xspace}
\newcommand{\np}{\textsf{NP}}
\newcommand{\p}{\textsf{P}}
\newcommand{\R}{\mathbb{R}}
\newcommand{\MEB}{{\sf MEB}}
\newcommand{\ext}{{\sf ext}}
\begin{document}

\title{{\bf FPT Constant-Approximations for Capacitated Clustering to Minimize the Sum of Cluster Radii}\footnote{Supported by the European Research Council (ERC) under the European Union’s
Horizon 2020 research and innovation programme (grant agreement no. 819416), and Swarnajayanti
Fellowship (no. DST/SJF/MSA01/2017-18).}}


%
\author[1]{Sayan Bandyapadhyay}
\author[2]{William Lochet}
\author[3]{Saket Saurabh}
\affil[1]{Department of Computer Science, Portland State University}
\affil[2]{LIRMM, Université de Montpellier, CNRS, Montpellier, France}
\affil[3]{The Institute of Mathematical Sciences, HBNI, Chennai, India}
\date{}

\maketitle
\thispagestyle{empty}
\begin{abstract}
    Clustering with capacity constraints is a fundamental problem that attracted significant attention throughout the years. In this paper, we give the first FPT constant-factor approximation algorithm for the problem of clustering points in a general metric into $k$ clusters to minimize the sum of cluster radii, subject to non-uniform hard capacity constraints (\radii). In particular, we give  a $(15+\epsilon)$-approximation algorithm that runs in $2^{\cO(k^2\log k)}\cdot n^3$ time. 

    \medskip
     When capacities are uniform, we obtain the following improved approximation bounds.
    
    \begin{itemize}
     \item A (4 + $\epsilon$)-approximation with running time $2^{\cO(k\log(k/\epsilon))}n^3$, which significantly improves over the FPT 28-approximation of Inamdar and Varadarajan [ESA 2020].
     
     \item A (2 + $\epsilon$)-approximation with running time $2^{\cO(k/\epsilon^2 \cdot\log(k/\epsilon))}dn^3$ and a $(1+\epsilon)$-approximation with running time $2^{\cO(kd\log ((k/\epsilon)))}n^{3}$ in the Euclidean space. Here $d$ is the dimension.

     \item A (1 + $\epsilon$)-approximation in the Euclidean space with running time $2^{\cO(k/\epsilon^2 \cdot\log(k/\epsilon))}dn^3$ if we are allowed to violate the capacities by (1 + $\epsilon$)-factor. We complement this result by showing that there is no (1 + $\epsilon$)-approximation algorithm running in time $f(k)\cdot n^{\cO(1)}$, if any capacity violation is not allowed. 
    \end{itemize}
\end{abstract}

\newpage 
\setcounter{page}{1}
\section{Introduction}
The \sumradii (clustering) problem is among the most popular and well-studied clustering models in the literature, together with $k$-center, $k$-means, and $k$-median \cite{DBLP:journals/jcss/CharikarP04, gonzalez1985clustering, AryaGKMMP-SIAMJ04,KanungoMNPSW04}. 
In \sumradii, we are given a set $P$ of $n$ points in a metric space with distance \textsf{dist}, and a non-negative  integer $k$ specifying the number of clusters sought. We would like to find: (i) a subset $C$ of $P$ containing $k$ points (called centers) and a non-negative integer $r_q$ (called radius) for each $q\in C$, and (ii) a function assigning each point $p\in P$ to a center $q\in C$ such that $\textsf{dist}(p,q)\le r_q$. The goal is to minimize the sum of the radii $\sum_{q\in C} r_q$. Alternatively, the objective is to select $k$ balls in the metric space centered at $k$ distinct points of $P$, such that each point $p\in P$ is contained in at least one of those $k$ balls and the sum of the radii of the balls is minimized.

In a seminal work, Charikar and Panigrahy~\cite{DBLP:journals/jcss/CharikarP04} studied the \sumradii problem. As mentioned in their paper, sum of radii objective  can be used as an alternative to the $k$-center objective to reduce the so called \emph{dissection effect}. The $k$-center objective is similar to sum of radii, except here one would want to minimize the maximum radius. As in $k$-center all balls are assumed to have the same maximum radius, the balls can have huge overlap. Consequently,  points that should have been assigned to the same cluster might end up in different clusters. This phenomenon is called the dissection effect which can be reduced by using the sum of radii objective instead.

Considering the sum of radii objective, Charikar and Panigrahy~\cite{DBLP:journals/jcss/CharikarP04} obtained a $3.504$-approximation running in polynomial  time, which is the best known approximation factor for this problem in polynomial time to date. Their algorithm is based on a primal-dual scheme coupled with an application of Lagrangean relaxation. Subsequently, Gibson et al.~\cite{DBLP:journals/algorithmica/GibsonKKPV10} obtained a $(1+\epsilon)$-approximation in quasi-polynomial time. It follows from the standard complexity theoretic assumptions that the problem cannot be \textsf{APX}-hard. We note that the problem is known to be \np-hard even in weighted planar metrics and metrics of constant doubling dimension \cite{DBLP:journals/algorithmica/GibsonKKPV10}. Surprisingly, the problem can be solved in polynomial time in Euclidean spaces when the dimension is fixed \cite{DBLP:journals/siamcomp/GibsonKKPV12}. When the dimension is arbitrary, one can obtain a $(1+\epsilon)$-approximation in $2^{\cO((k\log k)/\epsilon^2)}\cdot n^{O(1)}$ time, extending the coreset based algorithm for $k$-center \cite{Badoui2002}. 

\subsection{Our Problem and Results}
In this work, we are interested in the capacitated version of \sumradii.  
 Clustering with capacity constraints is a fundamental problem and has attracted significant attention recently~\cite{ByrkaRU16,ByrkaFRS15,CharikarGTS02,ChuzhoyR05,DemirciL16,Li15,Li17,DBLP:journals/mst/BhattacharyaJK18,DBLP:journals/algorithmica/DingX20,DBLP:conf/esa/AdamczykBMM019,DBLP:journals/corr/abs-1901-04628,DBLP:conf/icalp/Cohen-AddadL19}.
 Indeed, capacitated clustering is relevant in many real-life applications, such as load balancing where the representative of each cluster can handle the load of only a bounded number of objects. It is widely known that clustering problems become much harder in the presence of capacity constraints. 


Formally, in the \radii problem, 
along with the points of $P$ in a metric space, we are also given a non-negative integer $\eta_q$ for each $q\in P$, which denotes the capacity of $q$. The goal is similar to the goal of \sumradii except here each chosen center $q\in C$ can be assigned at most $\eta_q$ points of $P$. Alternatively, each cluster contains a bounded number of points specified with respect to the center of the cluster. In the uniform-capacitated version of the problem, $\eta_p=\eta_q$ for all $p,q\in P$, and we denote the capacity by $U$. We note that in this work, we only consider \emph{hard} capacities, i.e., each point can be chosen at most once to be a cluster center. In this setting, a major open question is to determine whether there is a polynomial time $\cO(1)$-approximation algorithm for \radii, even in the uniform-capacitated case. 

\begin{tcolorbox}
	\begin{description}
	\setlength{\itemsep}{-2pt}
	\item[Question $1$:] Does \radii admit 
	a polynomial time constant-approximation algorithm, even with uniform capacities? 
	\end{description}
\end{tcolorbox}

Designing polynomial time constant-approximations for capacitated clustering problems are notoriously hard. In fact such algorithms exist only for the $k$-center objective out of the four objectives mentioned before. For uniform capacitated $k$-center, Khuller and Sussmann~\cite{DBLP:journals/siamdm/KhullerS00} designed a 6-approximation improving a 10-approximation of Bar-Ilan et al.~\cite{barilan1993allocate} who introduced the problem. The first constant-approximation in the non-uniform case \cite{DBLP:conf/focs/CyganHK12} was designed after 12 years, which was subsequently improved to a 9-approximation by An et al.~\cite{DBLP:journals/mp/AnBCGMS15}. The capacitated problems with $k$-means and $k$-median objectives have attracted a lot of attention over the years. But, despite a recent progress for the uniform version in $\mathbb{R}^2$ \cite{DBLP:conf/soda/Cohen-Addad20}, where a PTAS is achieved, even in $\mathbb{R}^3$, the problem of finding a polynomial time constant-approximation remains open. The best-known polynomial time approximation factor in general metrics is $\cO(\log k)$ \cite{DBLP:conf/icalp/Cohen-AddadL19}, which is a folklore and is based on a tree embedding scheme. 

\medskip
\noindent
{\bf Technical Barriers for Sum of Radii.} The problem with sum of radii objective also appears to be fairly challenging. The main difficulty in achieving a polynomial time $\cO(1)$-approximation for \radii is obviously the presence of the capacity bounds even if they are uniform, which makes the problem resilient to the techniques used for solving \sumradii. The only polynomial time $\cO(1)$-approximation known for \sumradii is via a primal-dual scheme. However, it is not clear how to interpret the capacity constraints in the primal, in the realm of dual. Also, while the algorithms for capacitated $k$-center use LP-relaxation of the natural LP, the standard LP relaxation for \radii has a large integrality gap \cite{DBLP:conf/esa/0002V20}. Needless to say, the situation becomes much more intractable in the non-uniform capacitated case. 

\medskip
\noindent
{\bf Hardness of Approximation.} The lower bounds known for capacitated clustering are equally frustrating as their upper bounds. Surprisingly, the only known lower bounds are the ones for the uncapacitated versions, and hence trivially translated to the capacitated case. Due to the 20-year old work of Guha and Khuller \cite{guha1999greedy}, $k$-median and $k$-means are known to be NP-hard to approximate within factors of 1.735 and 3.943, respectively. In a recent series of papers, Cohen-Addad, Karthik and Lee~\cite{cohen2019inapproximability,cohen2021approximability,cohen2022johnson} have obtained improved constant lower bounds for various clustering problems in different metrics and settings. In particular, in the last work, they introduced an interesting Johnson Coverage Hypothesis \cite{cohen2022johnson} which helped them obtain improved bounds in various metrics. As mentioned before, \sumradii cannot be \textsf{APX}-hard, and hence there is no known inapproximability results that can be translated to the capacitated version. 

\medskip
In the light of the above discussions, one may conclude that the rather benign capacity constraints have played a bigger role compared to the choice of objective function, in our current lack of understanding of practical clustering models. It is necessary to resolve the question of existence of polynomial time $\cO(1)$-approximation for different objective functions to get a good grasp of this understanding. However, making any intermediate progress towards understanding capacitated clustering, irrespective of the objective function, is significant and timely.  

\medskip
\noindent
{\bf Coping with Capacitated Clustering.} In order to improve the understanding of these challenging open questions, researchers have mainly studied two types of relaxations to obtain constant-approximation algorithms. The more traditional approach taken for $k$-means and $k$-median is to design bi-criteria approximation where we are allowed to violate either capacity or the bound on the number of clusters by a small amount \cite{ByrkaRU16,ByrkaFRS15,CharikarGTS02,ChuzhoyR05,DemirciL16,Li15,Li17}. The other (relatively newer) approach is to design \emph{fixed-parameter tractable} (\FPT) approximation, thus allowing an extra factor $f(k)$ in the running time. We note that, in recent years, \FPT approximations are designed for classic problems improving the best known approximation factors in polynomial time, e.g., $k$-vertex separator \cite{lee2017partitioning}, $k$-cut \cite{gupta2018fpt} and $k$-treewidth deletion \cite{gupta2019losing}. 

\medskip
\noindent
{\bf \FPT Approximation for Clustering.}
In the context of clustering problems, the number of clusters $k$ is a natural choice for the parameter, as the value of $k$ is typically small in practice, e.g., $k\le 10$ in \cite{pedregosa2011scikit,steinbach2000comparison}. Consequently, the approach of designing \FPT approximation have become fairly successful for clustering problems and have led to interesting results which are not known or impossible in polynomial time. For example, constant-approximations are obtained for the capacitated version of $k$-median and $k$-means \cite{DBLP:conf/esa/AdamczykBMM019,DBLP:journals/corr/abs-1901-04628,DBLP:conf/icalp/Cohen-AddadL19}, which almost match the polynomial time constant approximation factors in the uncapacitated case. In the uncapacitated case of $k$-median and $k$-means, tight $(1.735+\epsilon)$ and $(3.943+\epsilon)$-factor \FPT approximations are recently obtained \cite{cohen2019tight,DBLP:conf/esa/AdamczykBMM019,DBLP:journals/corr/abs-1901-04628}, whereas the best known factors in polynomial time are only 2.611 \cite{byrka2014improved} and 6.357 \cite{ahmadian2019better}. These results are interesting in particular, as a popular belief in the clustering community is that there is no algorithmic separation between \FPT and polynomial time in general metrics (for example, see the comment in \cite{cohen2021approximability} after Theorem 1.3). We note that it is possible to obtain $(1+\epsilon)$-approximations in high-dimensional Euclidean spaces \cite{DBLP:journals/mst/BhattacharyaJK18,DBLP:journals/algorithmica/DingX20}, which is impossible in polynomial time, assuming standard complexity theoretic conjectures.     


\medskip
Inamdar and Varadarajan \cite{DBLP:conf/esa/0002V20} adapted the approach of designing \FPT approximation to study the \radii problem  with uniform capacities.  They make the first substantial progress in understanding this problem through the lens of fixed-parameter tractability. In particular, they obtained a 28-approximation algorithm for this problem that runs in time $2^{\cO(k^2)}n^{\cO(1)}$. Unfortunately, their algorithm does not work in the presence of non-uniform capacities. Based on their result, the following natural questions arise.


\vspace{2mm}

\begin{tcolorbox}
\begin{description}
\setlength{\itemsep}{-2pt}
\item[Question $2$:] Does \radii admit 
a constant-approximation algorithm, in \FPT\ time, even when  capacity constraints are non-uniform? 
\item[Question $3$:] Does \radii admit 
a $(1+\epsilon)$-approximation algorithm, in \FPT\ time, when the points are in  $\mathbb{R}^{d}$ (Euclidean Metric)? 
\end{description}
\end{tcolorbox}

We make significant advances towards answering Questions $2$ and $3$. Our first result completely answers Question $2$. 


\begin{restatable}{theorem}{nonuniform}
\label{thm:general_nonuniform}
	For any constant $\epsilon > 0$, the \radii problem admits a $(15+ \epsilon)$-approximation algorithm with running time $ 2^{\cO(k^2\log k)}\cdot n^3$.
\end{restatable}

Next, we consider the uniform-capacitated version and prove the following theorem significantly improving over the approximation factor of 28 in \cite{DBLP:conf/esa/0002V20}. 

\begin{restatable}{theorem}{uniform}\label{thm:general}
	For any constant $\epsilon > 0$, there exists a randomized algorithm for the \radii problem with uniform capacities that outputs with constant probability a $(4+ \epsilon)$-approximate solution in time  $2^{\cO(k\log (k/\epsilon))}\cdot n^3$.
\end{restatable}

The approximation factor in the above result is interesting in particular, as it almost matches the approximation factor of $3.504$ in the uncapacitated case and keeps the avenue of obtaining a matching approximation in polynomial time open.

Finally, we mention the Euclidean version of the problem where we show that adapting the standard coreset argument for regular $k$ clustering allows us to obtain the following two results.

\begin{restatable}{theorem}{twoapprox}
\label{thm:twoapprox}
	For any constant $\epsilon > 0$, there exists a randomized algorithm for the Euclidean version of \radii with uniform capacities that outputs with constant probability a $(2 + \epsilon)$-approximate solution in time $2^{\cO((k/\epsilon^2)\log (k/\epsilon))}\cdot dn^3$, where $d$ is the dimension. 
\end{restatable}

\begin{restatable}{theorem}{ptaskandd}
	\label{thm:ptaskandd}
	For any constant $\epsilon > 0$, the Euclidean version of \radii admits an  $(1+ \epsilon)$-approximation algorithm with running time $2^{\cO(kd\log ((k/\epsilon)))}n^{3}$.
\end{restatable}


We also complement our approximability results by hardness bounds. The \np-hardness of \radii trivially follows from the \np-hardness of \sumradii. We strengthen this bound by showing the following result. 

\begin{restatable}{theorem}{genhard}
 \radii with uniform capacities cannot be solved in $f(k)n^{o(k)}$ time for any computable function $f$, unless ETH is false. Moreover, it does not admit any FPTAS,  unless \p=\np. 
\end{restatable}

We also show an inapproximability bound in the Euclidean case even with uniform capacities. 

\begin{restatable}{theorem}{euclidhard}
\label{thm:euclidhard}
The Euclidean version of  \radii with uniform capacities does not admit any FPTAS even if $k=2$, unless \p=\np.
\end{restatable}

Although the above bound does not eradicate the possibility of obtaining a $(1+\epsilon)$-approximation in the Euclidean case, it shows that to obtain such an approximation, even when $k=2$, one needs $n^{f(\epsilon)}$ time for some non-constant function $f$ that depends on $\epsilon$. This is in contrast to the uncapacitated version of the problem, where one can get $(1+\epsilon)$-approximation in $2^{\cO((k\log k)/\epsilon^2)}\cdot n^{O(1)}$ time as mentioned before. 

As by products of our techniques we have also obtained improved bi-criteria approximations for the uniform-capacitated version of the problem where we are allowed to use $(1+\epsilon)U$ capacity. 


\begin{restatable}{theorem}{genbi}
 There is a randomized algorithm for \radii with uniform capacities that runs in time $2^{\cO(k \log (k/\epsilon))}\cdot n^{\cO(1)}$ and returns a solution with constant probability, such that each ball in the solution uses at most $(1+\epsilon)U$ capacity and the cost of the solution is at most $(2+\epsilon)\cdot  \emph{OPT}$, where \emph{OPT} is the cost of any optimal solution in which the balls  use at most $U$ capacity. 
\end{restatable}

The above theorem improves the approximation factor in Theorem \ref{thm:general}. In the Euclidean case, we obtain a similar algorithm. 

\begin{restatable}{theorem}{genbi}
 There is a randomized algorithm for the Euclidean version of \radii with uniform capacities that runs in time $2^{\cO((k/\epsilon^2)\log k)}\cdot dn^3$ and returns a solution with constant probability, such that each ball in the solution uses at most $(1+\epsilon)U$ capacity and the cost of the solution is at most $(1+\epsilon)\cdot  \emph{OPT}$, where \emph{OPT} is the cost of any optimal solution in which the balls  use at most $U$ capacity. 
\end{restatable}

Note that, by Theorem \ref{thm:euclidhard}, a result as in the above theorem is not possible if we are not allowed to violate the capacity.   

\medskip\noindent
{\bf Our Techniques.} 
To obtain the above results, we overcome the known hurdles for the sum of radii objective  in the presence of capacity constraints in a non-trivial fashion. The first hurdle one faces to approach \radii is how to deal with non-uniform capacities and different size balls at the same time. For example, it is not clear how to adapt the algorithm of Inamdar and Varadarajan \cite{DBLP:conf/esa/0002V20} for non-uniform capacities. Their algorithm works in levels. In each level, a certain number of balls of same radius $r_i$ are selected using a greedy strategy. All of these balls might not be chosen in the optimal
solution. But, as the capacities of all balls are same, it is possible to argue that an appropriate capacity reassignment is possible from the optimal balls to their solution. However, it is not clear how to do the capacity reassignment if all the capacities are not same. 

We employ a completely different approach. We try to find a solution where the balls are partitioned into two types: large and small. Each ball in our solution  represents a ball in the optimal solution (view as a bijection). Each large representative ball fully contains the  corresponding optimal ball. Thus the capacity reassignment of large balls is easy. Also, these large balls encompasses all the points. Now, each small representative ball has the property that this ball and the corresponding optimal ball are disjoint, and both are totally contained in the intersection of a subset of the large balls. Hence, one can show that the capacity reassignment of small balls is also possible. This gives us a valid solution and an assignment. It is not straightforward to see that such a well-structured solution exists. Our main contribution is to prove by construction the existence of such a solution, and the construction takes the desired FPT time.

To obtain the improved results for uniform capacities, we explore techniques beyond the literature of \sumradii. Among these, a technique which has been very successful in obtaining \FPT algorithms for other clustering problems is \emph{random sampling} \cite{DBLP:conf/icalp/Cohen-AddadL19,DBLP:journals/jacm/KumarSS10,DBLP:journals/mst/BhattacharyaJK18,DBLP:journals/algorithmica/DingX20,Badoui2002,Ding20}. However, all of these applications of the technique are problem specific and work mostly in the Euclidean spaces. In our work, we devise random sampling schemes for both general and Euclidean spaces which seem to work well with the sum of radii objective. We use these schemes to obtain improved constant-approximations in the uniform-capacitated case. Needless to say,  such a  uniform random sampling does not seem to be applicable in the non-uniform capacitated case. 

Following our work, Jaiswal et al.~\cite{DBLP:conf/innovations/Jaiswal0Y24} have obtained an improved approximation factor of $(4+\sqrt{13}+\epsilon)$ for \radii, where $\epsilon > 0$. They have also improved the factor in the uniform case to $(3+\epsilon)$. They pointed out an error in the proof of Lemma 2.5 in an earlier version of our paper. The current version contains the corrected proof.

\subsection{Related work}
{\bf \sumradii.} Due to its inherent popularity, \sumradii has been studied in other restricted settings as well where singleton clusters are not allowed and the metric is unweighted; polynomial time algorithms are devised in these settings \cite{DBLP:journals/algorithmica/BehsazS15,FriggstadJ22,DBLP:journals/dm/HeggernesL06}. Also, a facility location version of the problem is studied in dynamic setting where points arrive and disappear \cite{DBLP:journals/algorithmica/HenzingerLM20}. Similar to sum of radii, the sum of diameter \cite{DBLP:journals/njc/DoddiMRTW00} and the sum of $\alpha$-th power of radii (with $\alpha > 1$) \cite{DBLP:conf/isaac/BandyapadhyayV16} objectives are also studied in the literature.  

\smallskip
\noindent
{\bf $k$-center.} $k$-center is a closely related objective to sum of radii. The $k$-center problem admits polynomial time $2$-approximation \cite{gonzalez1985clustering,DBLP:journals/jacm/HochbaumS86}, and is hard to approximate in polynomial time within a factor of less than 2 \cite{DBLP:journals/sigact/Hochba97}, unless $\p=\np$. Badoui et al.~\cite{Badoui2002} designed a coreset based $(1+\epsilon)$-approximation for the Euclidean $k$-center that runs in $2^{\cO((k\log k)/\epsilon^2)}n^{\cO(1)}$ time. Subsequently, the dependency on $\epsilon$ was improved in \cite{DBLP:journals/comgeo/BadoiuC08}.   


\subsection{Preliminaries}

\noindent
{\bf \radii.} We are given a set $P$ of $n$ points in a metric space with distance \textsf{dist}, a non-negative integer $\eta_q$ for each $q\in P$, and a non-negative  integer $k$. The goal is to find: (i) a subset $C$ of $P$ containing $k$ points and a non-negative integer $r_q$ for each $q\in C$, and (ii) a function $\mu: P\rightarrow C$, such that for each $p\in P$, $\textsf{dist}(p,\mu(p))\le r_{\mu(p)}$, for each $q\in C$, $|\mu^{-1}(q)|\le \eta_q$, and $\sum_{q\in C} r_q$ is minimized. We will sometimes use OPT to denote the value of an optimal solution.

In the uniform-capacitated case, we denote the capacity by $U$. In the general metric version of our  problem, we assume that we are given the pairwise distances $\dist$ between the points in $P$. In the Euclidean version, $P$ is a set of points in $\mathbb{R}^{d}$ for some $d \ge 1$, $\dist$ is the Euclidean distance and any point in $\mathbb{R}^{d}$ can be selected as a center. Moreover, the capacities of all these centers are uniform.

We denote the ball with center $c$ and radius $r$ by $B(c,r)$. For any ball $B=B(c,r)$, we will use $\ext(B,r')$  to denote the ball $B(c,r + r')$. Sometimes we will also use $\rad(B)$ to denote the radius of $B$. Let $S$ be a set of points in $\mathbb{R}^{d}$. The \textit{minimum enclosing ball} of $S$, noted $\MEB(S)$ is the smallest ball in $R^d$ containing all points of $S$. We say a ball $B$ \emph{covers} a point $p$ if $p$ is in $B$.   

One important remark regarding solving our  capacitated clustering problem on $P$ is that, given a set of $k$ balls $\mathcal{B}$, the problem of deciding whether there is a valid assignment $\mu : P \rightarrow \mathcal{B}$ satisfying the capacities can easily be modeled as a bipartite matching problem. This implies in particular that if such an assignment exists, it can also be found in $\cO(\sqrt{k}n^{3/2})$ time. Therefore, in all our descriptions of the algorithms, we will focus on finding the solution balls while ensuring that a valid assignment exists. 

Another remark is that, in the case where every ball has capacity $U$, we can assume that $|P| \leq k \cdot U$, or the instance is a trivial NO instance.

Throughout the paper, by \emph{linear} running time, we mean a running time linear in  $n$, which may also depend on some polynomial function of $k$.

\medskip

\noindent 
{\bf Organization.} 
In Section \ref{sec:general} and \ref{sec:Euclidean}, we describe our  algorithms in general and Euclidean metrics, respectively. Section \ref{sec:hardness} contains the hardness results. Finally, in Section \ref{sec:conclude}, we conclude with some open questions. 

\section{{\radii}: General Metric}
\label{sec:general}

In this section, we describe our two approximation algorithms for \radii in any general metric.  
First we will present a $(15+\epsilon)$-approximation for \radii when we have non-uniform capacity constraints. This is followed by a $(4+\epsilon)$-approximation algorithm for \radii, when we have uniform capacity constraints for all the points. 

\subsection{\radii with Non-uniform Capacities: General Metric} 

In this subsection, we study the case of non uniform capacities. In this setting, for every point $x$ of $P$, there is an associated  integer $\eta_x$ and any ball centered at $x$ can be assigned at most $\eta_x$ points. The uniform case corresponds to the case where $\eta_x = U$ for all $x\in P$. For convenience,  we restate the theorem statement.

\nonuniform*

\vspace{2mm}

\begin{tcolorbox}[colback=white!5,colframe=white!75!black]
From now on, let $\mathcal{B}^\star := \{B^\star_1, \cdots, B^\star_k \}$ denote the set of balls of a hypothetical optimal solution, $\mu^\star : P \rightarrow \mathcal{B}^\star$ be the associated assignment and for all $i \in [k]$, let $r^\star_i$  and $c^\star_i$ denote the radius and the center of the ball $B^\star_i$, respectively. By Lemma \ref{lem:guess_radius_general}, just bellow, we can assume that the algorithm knows an approximate radius $r_i$ for each $r^\star_i$. For a ball  $B^\star_i \in \mathcal{B^\star}$, we say that a ball $B_i$ is an {\em approximate} ball of $B^\star_i$ if $B^\star_i \subseteq B_i$, and if $x_i$ denotes the center of $B_i$, then {\em $\eta_{x_i} \geq |(\mu^\star)^{-1}(B^\star_i)|$}. Note that we can then associate $(\mu^\star)^{-1}(B^\star_i)$ to $B_i$. 
\end{tcolorbox}
\vspace{2mm}

Let us first show that it is possible to guess an approximation to each of the radii, $r_i$, in polynomial time.

\begin{lemma}\label{lem:guess_radius_general}
	For every $0< \epsilon <1 $, there exists a randomized algorithm, running in linear time that finds with probability at least $\frac{\epsilon^k}{ k^k\cdot n^2}$ a set of reals $\{r_1, \dots, r_k \}$ such that for every $i \in [k]$,  $r^{\star}_i \leq r_i$and $\sum_{i \in [k]}r_i \leq (1+ \epsilon) \sum_{i \in [k]}r^\star_i$.
\end{lemma}

\begin{proof}
	Without loss of generality, we  assume that $r^\star_1 \geq r^\star_2 \geq \dots r^\star_{k-1} \geq r^\star_k$. Since $r^\star_1$ corresponds to the distance between two elements of $P$, there are at most $n^2$ possible choices for $r^\star_1$. Therefore with probability at least $1/n^2$, the algorithm {\em first} can choose this value uniformly at random. Assume this is the case and let $r_1 = r^\star_1$. Knowing $r_1 = r^\star_1$ exactly, knowing the value of each $r^\star_i$ up to a precision of $\frac{\epsilon r_1}{k}$ would yield the desired result. However, since $r^\star_1$ is the largest radius, there are only $r^\star_i \cdot \frac{k}{\epsilon r_1} \leq \frac{k}{\epsilon} $ choices for each $r^\star_i$, which gives $(\frac{k}{\epsilon})^k$ possible total choices. This concludes the proof. 
\end{proof}

From now on we {\em assume for simplicity that the algorithm knows an approximate value $r_i$ of $r^\star_i$ for all $i \in [k]$.} Let us give some informal ideas about how the algorithm of Theorem \ref*{thm:general_nonuniform} works. Some technicalities, especially about making sure we don't pick the same center twice, will be left out to the more formal description of the algorithm.

\subsubsection*{Informal sketch}

Ideally we would like to find for each optimal ball $B^\star_i$ an approximate ball $B_i$ having the same center as $B^\star_i$ and a radius $r_i \leq C \cdot r^\star_i$, for some constant $C$. Indeed, if we have such a set of balls, then the obvious assignment $\mu$ defined as $\mu(x) = B_i$ whenever $\mu^\star(x) = B^\star_i$ would give a solution. While this is not possible in general, the algorithm will start with a greedy procedure to get a set of approximate balls $\mathcal{B}_1$ for some indices $I_1$. The procedure is quite simple: start with  $\mathcal{B}_1 := \emptyset$, $I_1 = \emptyset$ and as long as the union of balls in $\mathcal{B}_1$ does not cover $P$, pick a point $x$ of $P$ not in the union, guess the index $i$ such that $\mu^\star(x) = B^\star_i$ and pick $c$ the point at distance at most $r_i$ of $x$ which maximises the value of $\eta_c$. Since $c^\star_i$ is at distance at most $r_i$ of $x$, we have that $\eta_{c^\star_i} \leq \eta_c$ and that $\textsf{dist}(c,c^\star_i) \leq 2r_i$, which means that the ball $B_i$ of radius $5r_i$ centered around $c$ is an approximate ball of $B^\star_i$. Therefore, the algorithm will add $B_i$ to $\mathcal{B}_1$ and the index $i$ to $I_1$. This procedure stops when the union of $\mathcal{B}_1$ covers $P$. At the end of that first step, we have that $\mathcal{B}_1$ contains an approximate ball for each ball $B^\star_i$ of radius $5r_i$ with $i \in I_1$. And while the union of $\mathcal{B}_1$ covers all $P$, we are far from being done as the capacity constraints have not been taken into account.

Consider now a ball $B^\star_j$ such that $j \not \in I_1$, which means that no approximate ball of $B^\star_j$ is in $\mathcal{B}_1$. In the best case (Lemma \ref{lem:nonuni_second}), there is a ball $B_i \in \mathcal{B}_1$ of center $x_i$ approximating $B^\star_i$ such that $5r_i \leq r_j$ and $B_i \cap B^\star_j$ is non empty. Indeed, in that case the ball of radius $5 r_i + r_j$ around $x_i$ contains $x^\star_j$, the center of $B^\star_j$. This means that if $x$ is the point in that ball maximizing $\eta_x$, then the ball of center $x$ and radius $2\cdot(5r_i + 2r_j)  \leq 4 r_j$ contains $c^\star_j$ and thus the ball $B_j$ of center $x$ and radius $2\cdot(5r_i + 2r_j) + r_j \leq 5 r_j$ contains $B^\star_j$ and is an approximate ball of $B^\star_j$. Therefore, if such indices $j$ and $i$ exist, the algorithm can guess them and add a new approximate ball to $\mathcal{B}_1$.

After this second step, we reach a point where, if $j \not \in I_1$ and $i \in I_1$ are such that $B^\star_j \cap B_i \not = \emptyset$, then $r_j \leq 5 r_i$. In particular, incurring an extra $5r_i$ as we just did to get a replacement for $c^\star_j$ is too costly. For this reason, at this step of the algorithm we stop trying to find approximate balls and instead focus on finding balls to "fix" the assignment. Since the balls in $\mathcal{B}_1$ are approximate balls, it means that we can replace $B^\star_i$ with $B_i \in \mathcal{B}_1$ for any $i \in I_1$ (and take the other $B^\star_j$), and still have a solution to our problem with slightly bigger balls. Abusing notation we can still use $\mu^\star$ for the valid assignment. Now for an index $j \not \in I_1$, the ball $B^\star_j$ intersects a subset, say $T_j$, of balls in $\mathcal{B}_1$. Ideally we would like to find a ball $B_j$ of center $x$ and radius $r_j$ such that $\eta_x \geq \eta_{c^\star_j}$ and $|B_j \cap (\mu^\star)^{-1}(B_i)| \geq |(\mu^\star)^{-1}(B^\star_j) \cap B_i|$ for all $i \in T_j$. Indeed, in that case we could replace $B^\star_j$ by $B_j$ and the condition on $B_j \cap B^\star_i$ ensures that we could adapt the assignment $\mu^\star$ to be a valid assignment by assigning $(\mu^\star)^{-1}(B^\star_j) \cap B_i$ to $B_i$ and a set of the same size in $B_j \cap (\mu^\star)^{-1}(B_i)$ to $B_j$. 

The main difficulty here is that even if we guess the set $T_j$, picking $B_j$ greedily is not possible as there might be some competitions between the sizes of the intersection with the different balls in $B_i$ for $i \in T_j$ (we cannot afford to guess the $|(\mu^\star)^{-1}(B^\star_j) \cap B_i|$). The way to avoid this problem is to expand all the balls $B_i$ of $\mathcal{B}_1$ by $10r_i$. Indeed, since we have assumed that $r_j \leq 5 r_i$ for every $i \in T_j$, it means now that $B^\star_j$ is \textbf{entirely contained} in the $\ext(B_i, 10r_i)$ (expanded ball) for $i \in T_j$. So now, denoting $P_j$ to be the intersection of all the $\ext(B_i, 10r_i)$ for $i \in T_j$ where we removed all the other $B_{i'}$ for $i' \in I_1 \setminus T_j$, we have that $(\mu^\star)^{-1}(B^\star_j)$ is a subset of $P_j$, and we can take $B_j$ as the ball of center $x$ and radius $r_j$ which maximizes $s_x = \min\{\eta_x, |B(x,r_j) \cap P_i|\}$. Here there are some technicalities if $B_j$ intersects some ball $B^\star_{j'}$ for $j' \not \in I_1$ (that includes $j = j'$), but assume for the moment that it is not the case and let us hint why we can actually replace $B^\star_j$ by $B_j$ in our solution if all the balls $B_i \in \mathcal{B}_1$ are replaced by $\ext(B_i, 10r_i)$. Indeed, by choice of $B_j$, we can assign $s_x$ points of $P_j$ to $B_j$. By our assumptions, all these points were assigned to $B_i$ for $i \in T_j$ in $\mu^\star$, which means that by assigning these points to $B_j$, there is now a new budget $s_x$ of available points in the union of $\ext(B_i, 10r_i)$ for $i \in T_j$. However, since $B^\star_j$ is entirely contained in the $\ext(B_i, 10r_i)$ for $i \in T_j$, and by choice of $s_x$, we can assign all the elements of $(\mu^\star)^{-1}(B^\star_j)$ to the balls $\ext(B_i, 10r_i)$ for $i \in T_j$ using this new budget. 

Therefore, the last phase of the algorithm consists in building a set of ``replacement" balls $\mathcal{B}_2$ for the balls $B^\star_j$ with $j \not \in I_1$ by guessing $T_j$ and building the intersection $P_j$ to take greedily a ball of radius $r_j$ inside that set (see Lemma \ref{lem:nonuni_third}). This is done sequentially, and the set $I_2$ will contain all indices $j$ for which  $\mathcal{B}_2$ contains a replacement ball $B_j$ for $B^\star_j$. An important remark here is that the properties required for balls in $\mathcal{B}_2$ are dependent on the balls in $\mathcal{B}_1$ and not just the optimal balls. For technical reasons, we might have to add a new ball in $\mathcal{B}_1$ during the process of building  $\mathcal{B}_2$, in which case we cannot guarentee that the properties of balls in $\mathcal{B}_2$ hold anymore. If this happens, the algorithm will then erase all the choices of $\mathcal{B}_2$ and $I_2$ and start the second phase again. However, since we only do this when $I_1$ gets bigger, this is done at most $k$ times before $I_1 = [k]$. The algorithm stops when the sets $I_1$ and $I_2$ contains all indices of $[k]$ which means that each ball $B^\star_j$ either has an approximate ball in $\mathcal{B}_1$ or a replacement ball in $\mathcal{B}_2$.  

\subsubsection*{The algorithm}

As explained previously, the algorithm maintains two disjoint sets of indices $I_1$ and  $I_2$, initially set to $\emptyset$, as well as two sets of balls $\mathcal{B}_1$ and $\mathcal{B}_2$, also initially set to $\emptyset$. $\mathcal{B}_1$ and $\mathcal{B}_2$ will eventually contain a representative ball $B_i$ for every ball $B^\star_i$ in the optimal solution. Moreover, we will argue that there exists a valid assignment of the points to the balls (with an expansion) in the union of these two sets.  

For every $i \in I_2$, let $T_i$ denote the set of indices $j$ of $I_1$, such that $B^\star_i \cap B_j$ is not empty, and $P_i$ denote the intersection of the extensions $\ext(B_j, 10r_j)$ over all $j \in T_i$ after removing the points of $B_s$ for $s \in I_1 \setminus T_i$.

We say that  the sets $(I_1, I_2,\mathcal{B}_1,\mathcal{B}_2)$ form a \textit{valid configuration} if the following properties are satisfied. 

\begin{itemize}
\setlength{\itemsep}{-1pt}
	\item For every $i \in I_1$, there is an approximate ball $B_i \in \mathcal{B}_1$ of $B^\star_i$ of radius at most $5r_i$. 
	\item For every $i \in I_2$, 
	$T_i$ is non-empty, $B^\star_i \subseteq P_i$ and there exists a ball $B_i \in \mathcal{B}_2$ of center $x_i$ and radius $r_i$, such that both $\eta_{x_i}$ and $B_i \cap P_i$ have size at least $|(\mu^\star)^{-1}(B^\star_i)|$. 
	
	\item For $i,j \in I_2$, $B_i$ and $B_j$ do not intersect. 
	\item For every $i \in I_2$ and $s \in [k] \setminus I_1$, $B^\star_s$ and $B_i$ do not intersect. 
	\item For every $j \in [k] \setminus I_1$, then $c^\star_j$ cannot be the center of a ball in $\mathcal{B}_1$.
\end{itemize} 

Before describing the algorithm to construct a valid configuration $(I_1, I_2,\mathcal{B}_1,\mathcal{B}_2)$ such that $I_1 \cup I_2 = [k]$, let us show that such a configuration would indeed yield an approximate solution.

\begin{lemma}\label{lem:nonuni_approx}
	Let $(I_1, I_2,\mathcal{B}_1,\mathcal{B}_2)$ be a valid configuration such that $I_1 \cup I_2 = [k]$, then the set of balls $\mathcal{B}$ containing the balls in $\mathcal{B}_2$, as well as for every $i \in I_1$ the ball $\ext(B_i, 10r_i)$ is a $15$-approximate solution. 
\end{lemma}

\begin{proof}
The fact that the sum of radii of the balls in $\mathcal{B}$  is at most $15$ times the optimal solution  follows from the 
definition. To prove the lemma, we have to show that this is a valid solution by giving a valid assignment. Recall that $\mu^\star$ is the assignment for $\mathcal{B^\star}$.  

For every $i \in I_2$, recall  that $T_i$ denotes the set of indices $j$ of $I_1$ such that $B^\star_i$ intersects $B_j$ and $P_i$ denotes the intersection of all the $\ext(B_j, 10r_j)$ for $j \in T_i$ where we removed the points in $B_s$ for $s \in I_1 \setminus T_i$
for $j \in I_1$. By definition of a valid configuration, if we use $Y_i$ to denote the set $(\mu^\star)^{-1}(B^\star_i)$, then there exists a set $X_i$ of size $|Y_i|$ in $B_i \cap P_i$.  The following claim is a crucial part of the proof. 

\begin{claim}
Any point $x \in X_i$ is such that $\mu^\star(x) = B^\star_j$ for some $j  \in T_i$
\end{claim}

\begin{proof}
Indeed, by definition, the only balls $B_j \in \mathcal{B}_1$ containing an element $x$ of $P_i$ are such that $j \in T_i$. In particular if $j \in I_1 \setminus T_i$, then $B^\star_j \subseteq B_j$ and therefore does not contain $x$. Moreover, if $j \in I_2$, then we know by definition that $B^\star_j \cap B_i$ is empty (that includes $B^\star_i$). Since $x \in B_i$, this concludes the proof of our claim. 
\end{proof}

The previous claim implies that, if we define for every $j \in T_i$ the set $X_{i,j} := X_i \cap   (\mu^\star)^{-1}(B^\star_j)$, then the $X_{i,j}$ for $j \in T_i$ actually define a partition of $X_i$. As $|X_i| = |Y_i|$, we can partition the set $Y_i$ into sets $Y_{i,j}$ for $j \in T_i$ such that $|X_{i,j}| = |Y_{i,j}|$ for all $j \in T_i$. Remember that, for $j \in T_i$,  $Y_i \subseteq  \ext(B_j, 10r_j)$. By convention, if $j \in T_1 \setminus T_i$, then $X_{i,j}$ and $Y_{i,j}$ are defined as the empty set. 

For every $j \in I_1$, we can now define $L_j = \left( (\mu^\star)^{-1}(B^\star_j) \setminus (\cup_{i \in I_2} X_{i,j}) \right) \cup_{i \in I_2} Y_{i,j}$. Since the sets $B_i$, for $i \in I_2$, are pair-wise disjoint and $|X_{i,j}| = |Y_{i,j}|$ for all elements $j \in I_1$ and $i \in I_2$, we have that $|L_j| = |(\mu^\star)^{-1}(B^\star_j)|$. Moreover, since $Y_{i,j}$ is non empty only if $Y_i \subseteq \ext(B_j, 10r_j)$, it means that $L_j \subseteq  \ext(B_j, 10r_j)$ and because $B_j$ is an approximate ball of $B^\star_j$, it means that the center $x_j$ of $B_j$ is such that $\eta_{x_j} \geq |(\mu^\star)^{-1}(B^\star_j)| = |L_j|$.
	
	Finally this means that the function $\mu$ such that $\mu^{-1}(\ext(B_j, 10r_j)) = L_j$ for all $j \in I_1$ and $\mu^{-1}(B_i) = X_i$ for all $i \in I_2$ is a valid assignment from $P$ to $\mathcal{B}$, which completes the proof.
\end{proof}

Now, we describe the algorithm that constructs the desired configuration. The first phase of the algorithm will consist of a greedy selection of elements in $I_1$, such that the union of $\mathcal{B}_1$ covers $P$ (Lemma \ref{lem:nonuni_first}). As said previously, this will not imply that we can assign points to these balls, without violating capacity constraints.  
The following two other lemmas (Lemma \ref{lem:nonuni_second} and \ref{lem:nonuni_third}) will be used to achieve that.

As we deal with hard capacities, we cannot reuse any center. We need the following definition to enforce that. We call a point $p\in P$ an \textit{available center},  
if $p$ has not already been selected as a center of a ball in $\mathcal{B}_1$ or $\mathcal{B}_2$.

\begin{lemma}\label{lem:nonuni_first}
	If $(I_1, I_2 = \emptyset,\mathcal{B}_1,\mathcal{B}_2= \emptyset)$ is a valid configuration such that the union of the balls in $\mathcal{B}_1$ do not cover $P$, then there exists a randomized algorithm, running in linear time and with probability at least $1/2k^2$, that finds an index $s$ and a ball $B_s$ such that adding $s$ to $I_1$ and $B_s$ to $\mathcal{B}$ still yields a valid configuration.
\end{lemma}

\begin{proof}
	Let $x$ be any point in $P$ not covered by the union of the balls in $\mathcal{B}_1$ and $i$ be the index such that $\mu^\star(x) = B^\star_i$. Let $c$ be the available potential center in $P$ at distance at most $r_i$ from $x$ which maximises the value of $\eta_c$. If $c$ is not a center of some $B^\star_j$ for $j \not \in (I_1 \cup I_2 )$, then the ball $B_i$ of center $c$ and radius $3r_i$ is an approximate ball of $B^\star_i$ and thus adding $B_i$ to $\mathcal{B}_1$ and $i$ to $I_1$ yields a valid configuration. If $c$ is a center of some $B^\star_j$ for $j \not \in (I_1 \cup I_2 )$, then adding the ball $B_j$ of center $c$ and radius $r_j$ to $\mathcal{B}_1$ and $j$ to $I_1$ also yields a valid configuration. 

	The algorithm will then pick uniformly at random an index $i' \in [k]$, then decide uniformly at random whether the available center $c'$ at distance at most $r_{i'}$ is a center of some $B^\star_j$ for some $j \not \in (I_1 \cup I_2 )$. If it decides negatively, then the algorithm will output $s := i'$ and $B_s$ the ball of center $c'$ and radius $3r_{i'}$. If it decides positively, then the algorithm will then also pick uniformly at random and index $j' \in [k]$ and output $s := j'$ as well as $B_s$ the ball of center $c'$ and radius $r_{j'}$. 

	The algorithm then suceeds if $i' = i$, if it decides correctly if $c'$ is a center of some $B^\star_j$ and if $j' = j$ in the case where it is. Overall this is true with probability at least $\frac{1}{k \cdot 2 \cdot k}$, which ends the proof. 
\end{proof}

The first phase of the algorithm  consists of applying the algorithm from Lemma~\ref{lem:nonuni_first}  until the union of the balls in $\mathcal{B}_1$ covers all the points in $P$. The next two lemmas are used in the next phase of the algorithm. 

\begin{lemma}\label{lem:nonuni_second}
	If $(I_1, I_2=\emptyset,\mathcal{B}_1,\mathcal{B}_2=\emptyset)$ is a valid configuration such that the balls in $\mathcal{B}_1$ cover the points of $P$ and there exist two indices $i \in I_1$ and $j \in [k] \setminus (I_1 \cup I_2)$ such that $B_i$ and $B^\star_j$ intersect and $r_j \geq 5 r_i$, then there exists a randomized algorithm that in linear time and with probability at least $1/2k$, finds an index $t \in [k] \setminus (I_1 \cup I_2)$ and a ball $B_t$ such that $(I_1 \cup \{t\}, I_2,\mathcal{B}_1 \cup \{B_t\},\mathcal{B}_2)$ is a valid configuration.
\end{lemma}

\begin{proof}
	Let $x_i$ denote the center of $B_i$, and $B'$ be the ball of center $x_i$ and radius $5r_i + r_j$. Because $B_i$ is an approximate ball of $B^\star_i$, and $B^\star_j$ and $B_i$ intersect, we have that $B'$ contains $c^\star_j$. Let $x$ be the potential center of $B'$ which maximises $\eta_x$. If there exists an index $j' \in [k] \setminus (I_1 \cup I_2)$, such that the ball $B^\star_{j'}$ is centered at $x$, then $t := j'$ and the ball $B_{t}$ of center $x$ and radius $r_{j'}$ satisfy the property of the lemma (remember that $\mathcal{B}_2=\emptyset$). If not, then the ball $B_j$ at center $x$ and of radius $2 \cdot (5r_i + r_j) + r_j$ is an approximate ball of $B^\star_j$ of radius at most $5r_j$. Indeed, because the ball at  center $x_i$ and of radius $(5r_i + r_j)$ contains both $x$ and $c^\star_j$, it means that the ball at center $x$ and of radius $ 2(5r_i + r_j)$ contains $c^\star_j$ and thus $ B^\star_j \subseteq B_j$. Again, as $\mathcal{B}_2=\emptyset$, then $t := j$ and $B_t := B_j$ satisfy the properties of the lemma.
	
	The algorithm therefore consists of choosing uniformly at random which of the two cases is true. In the first case it also chooses uniformly at random an index $j_1 \in [k]$ and outputs $t:= j_1$ as well as the ball $B_t$ of center $x$ and radius $r_{j_1}$. In the second case it outputs $t := j$ as well as the ball $B_t$ of center $x$ and radius $2 \cdot (5r_i + r_j) + r_j$. The previous discussion implies that the algorithm succeeds if it chooses correctly between the two cases, and in the first case if $j_1 = j'$. Overall, the probability of success is at least $1/2k$. 
\end{proof}

\begin{lemma}\label{lem:nonuni_third}
	Suppose $(I_1, I_2,\mathcal{B}_1,\mathcal{B}_2)$ is a valid configuration with the property that the balls in $\mathcal{B}_1$ cover the points of  $P$ and for every $i \in [k] \setminus(I_1 \cup I_2)$ and $j \in I_1$, such that $B_j$ and $B^\star_i$ intersect, $r_i \leq 5 r_j$. Then there exists a randomized algorithm that in linear time and with probability at least $1/6(k+1)^2$, either finds an index $t \in [k] \setminus I_1$ and a ball $B_t$ such that $(I_1 \cup \{t\}, I_2=\emptyset,\mathcal{B}_1 \cup \{B_t\},\mathcal{B}_2=\emptyset)$ is a valid configuration, or finds an index $s \in [k] \setminus (I_1 \cup I_2)$ and a ball $B_s$ such that $(I_1, I_2 \cup \{s\},\mathcal{B}_1,\mathcal{B}_2 \cup \{B_s\})$ is a valid configuration.
\end{lemma}

\begin{proof}
    Let $i$ be an element of $[k] \setminus(I_1 \cup I_2)$ and let $T_i$ denote the set of indices $j \in I_1$ such that $B_j \cap B^\star_i$ is non-empty. By the hypothesis of the lemma, we have that $r_i \leq 5 r_j$ for each element $j \in T_i$. In particular, it means that $B^\star_i \subseteq \ext(B_j,10r_j)$ for every $j \in T_i$. Let $P_i$ denote the intersection of all those sets $ \ext(B_j,10r_j)$ where we removed $B_s$ for all $s \in I_1 \setminus T_i$. By definition, we have that $(\mu^\star)^{-1}(B^\star_i) \subseteq P_i$ and moreover, as $(I_1, I_2,\mathcal{B}_1,\mathcal{B}_2)$ is a valid configuration, $c^\star_i$ is an available center in $P_i$ and $B^\star_i$ do not intersect any ball $B_j$ with $j \in I_2$. 

    We will now construct the set of points $x_1, \dots, x_{\ell}$, as well as a set of balls of radius $r_i$ $B_{x_1}, \dots, B_{x_{\ell}}$ inductively as follows. Let $x_1$ be the available center in $P_i$ such that, denoting $B_{x_1}$ the ball at center $x_1$ and of radius $r_i$, $B_{x_1}$ is disjoint from all the elements in $\mathcal{B}_2$ and $s_{x_1}=\min\{\eta_{x_1}, |B_{x_1} \cap P_i|\}$ is maximized. Now for any $j \in [k]$, assuming $x_1, \dots, x_j$ as well as the balls $B_{x_1}, \dots B_{x_j}$ have been constructed, we define $x_{j+1}$ to be the available center in $P_i$ such that $x_{j+1}$ is at distance at least $4r_i$ from any other $x_s$ with $s \in [j]$, and denoting $B_{x_{j+1}}$ the ball at center $x_{j+1}$ and of radius $r_i$, $B_{x_{j+1}}$ is disjoint from all the elements in $\mathcal{B}_2$ and $s_{x_{j+1}}=\min\{\eta_{x_{j+1}}, |B_{x_{j+1}} \cap P_i|\}$ is maximized. The process stops when either $\ell = k+1$, or all the available centers of $P_i$ are at distance less than $4r_i$ from $x_1, \dots, x_{\ell}$. 

    Suppose first that an index $s \in [k+1]$ exists such that $c^\star_i$ is at a distance less than $4r_i$ from $x_s$. Without loss of generality, let us assume that $s$ is the smallest such index. In particular, for every $s' <s$, $c^\star_i$ is at a distance at least $4r_i$ from $x_{s'}$. Moreover, we know that $B^\star_i \subseteq T_i$, which means that by choice of $x_s$, $s_{x_s} \geq |(\mu^\star)^{-1}(B^\star_i)|$. This means that the ball $B_i$ of center $x_s$ and radius $5r_i$ is an approximate ball of $B^\star_i$, and thus, if $x_s$ is not a center $c^\star_j$ for some $j \in [k] \setminus I_1$ different from $i$, then $(I_1 \cup \{i\}, I_2=\emptyset,\mathcal{B}_1 \cup \{B_{i}\},\mathcal{B}_2=\emptyset)$ is a valid configuration, in particular as $\mathcal{B}_2=\emptyset$. If $x_s$ is a center $c^\star_j$ for $j \in [k] \setminus I_1$, then denoting $B_j$ the ball of center $x_s$ and radius $r_j$, we have that $(I_1 \cup \{j\}, I_2=\emptyset,\mathcal{B}_1 \cup \{B_{j}\},\mathcal{B}_2=\emptyset)$ is a valid configuration. From now on we will assume that this case does not occur. In particular, it means that $\ell = k+1$.  

    Suppose now that there exists an index $i' \in [k] \setminus I_1$ such that $r_{i'} \geq r_i$ and $B^\star_{i'}$ intersects $B_{x_s}$ for some $s \in [k+1]$. In that case, the ball of center $x_s$ and radius $r_i + r_{i'}$ contains $c^\star_{i'}$. Therefore, if $x$ denotes the available center inside that ball that maximises $\eta_x$, we have $\eta_x \geq \eta_{c^\star_{i'}}$. Finally, since $r_i \leq r_{i'}$, this means that the ball $B_{i'}$ of center $x$ and radius $2(r_i + r_{i'}) + r_{i'} \leq 5 r_{i'}$ is an approximate ball of $B^\star_{i'}$. There might be several options for $i'$, but we can make an arbitrary choice. Thus, if $x$ is not a center $c^\star_j$ for some $j \in [k] \setminus I_1$ different from $i'$, then $(I_1 \cup \{i'\}, I_2=\emptyset,\mathcal{B}_1 \cup \{B_{i'}\},\mathcal{B}_2=\emptyset)$ is a valid configuration. If $x$ is a center $c^\star_j$ for $j \in [k] \setminus I_1$, then denoting $B_j$ the ball of center $x$ and radius $r_j$, we have that $(I_1 \cup \{j\}, I_2=\emptyset,\mathcal{B}_1 \cup \{B_{j}\},\mathcal{B}_2=\emptyset)$ is a valid configuration.

    Suppose now that for every index $i' \in [k] \setminus I_1$ such that $r_{i'} \geq r_i$, we have that $B^\star_{i'}$ is disjoint from every ball in $B_{x_1}, \dots, B_{x_{k+1}}$ and $c^\star_i$ is at distance $4r_i$ from every $x_s$ with $s \in [k+1]$. For every other index $i'' \in [k] \setminus I_1$, we know that $B^\star_{i''}$ can only intersect a single ball of $B_{x_1}, \dots, B_{x_k+1}$. Indeed, all the $x_1, \dots, x_{k+1}$ are at distance at least $4r_i$ of one another by definition and $r_{i''} < r_i$. In particular this means that there is an index $s \in [k+1]$ such that no ball $B^\star_{i'}$ intersects $B_{x_s}$ for $i \in [k] \setminus I_1$. In particular, it means that $x_s$ is different from all the centers $c^\star_j$ with $j \in [k] \setminus I_1$. Moreover, when $x_s$ and $B_{x_s}$ were chosen, we know that  $c^\star_i$ is an available center in $P_i$ at distance $4r_i$ from the points $x_1, \dots, x_{s-1}$. This means in particular that $\eta_{x_s}$ and $|B_x \cap P_i|$ are both larger than $|(\mu^\star)^{-1}(B^\star_i)|$. Moreover, we know by choice that $B_{x_s}$ doesn't intersect any element of $\mathcal{B}_2$, and thus $(I_1, I_2 \cup\{i\},\mathcal{B}_1, \mathcal{B}_2 \cup \{B_{x_s}\})$ is a valid configuration.

    Finally, the algorithm consists of constructing the $x_1, \dots, x_{\ell}$ and $B_{x_1}, \dots, B_{x_{\ell}}$ and deciding between the three cases randomly, and in the last two cases picks uniformly the index $i'$ and then outputs the described ball and index. In the last two cases, it also needs to decide if there exists $i''$ such that the potential center considered is actually $c^\star_{i''}$ and in which case pick that index uniformly at random as well. Overall, the probability of success of this algorithm is at least $1/6(k+1)^2$.  
\end{proof}

We are now ready to prove our main theorem.

\begin{proof}[Proof of Theorem~\ref{thm:general_nonuniform}]
Let us describe the algorithm. First, it applies Lemma \ref{lem:guess_radius_general} to obtain an approximation $r_i$ of each $r^\star_i$ with probability at least $\frac{\epsilon^k}{k^k\cdot n^2}$. Then the algorithm initialize a valid configuration $(I_1 = \emptyset, I_2 = \emptyset, \mathcal{B}_1 = \emptyset, \mathcal{B}_2= \emptyset )$ and run the algorithm of Lemma \ref{lem:nonuni_first} at most $k$ times, until $\mathcal{B}_1$ covers all the points of $P$. At each step, the probability of success is at least $1/k^2$, so in total at least $1/k^{2k}$. 
Then the algorithm enters into the second phase. This phase is divided into multiple steps. 

In the beginning of each step, we maintain the invariant that the current configuration is valid with $I_2 = \emptyset$ and $\mathcal{B}_2= \emptyset$. Each step then consists of a series of applications of Lemma  \ref{lem:nonuni_second} followed by a series of applications of Lemma  \ref{lem:nonuni_third}. The current step ends when an index is added to $I_1$ by the application of Lemma \ref{lem:nonuni_third}, and hence at that point $I_2 = \emptyset$ and $\mathcal{B}_2= \emptyset$, or $I_1 \cup I_2 = [k]$. We go to the next step (maintaining the invariant), unless $I_1\cup I_2=[k]$ in which case the algorithm terminates. 

In a step, the algorithm decides which lemma to apply as long as $I_2=\emptyset$. Otherwise, it applies only Lemma \ref{lem:nonuni_third}. If $I_2=\emptyset$, it randomly decides if  there exists indices $i \in I_1$ and $j \in [k] \setminus (I_1 \cup I_2)$ such that $B^\star_i$ and $B^\star_j$ intersect and $r_j \geq 5r_i$. In which case, the algorithm applies Lemma \ref{lem:nonuni_second} to increase the size of $|I_1|$ in linear time and with probability at least $1/2^2k$. If no such pair of indices exists, then the algorithm applies Lemma \ref{lem:nonuni_third} to increase $|I_1 \cup I_2|$ or $|I_1|$ in linear time and with probability at least $1/6(k+1)^2$. 

\begin{claim}
Assuming the algorithm made all the correct random choices, it terminates with a valid configuration $(I_1, I_2,\mathcal{B}_1,\mathcal{B}_2)$ such that $I_1 \cup I_2 = [k]$ after $O(k^2)$ applications of Lemma \ref{lem:nonuni_second} and \ref{lem:nonuni_third}.  
\end{claim}

\begin{proof}
First, we argue about the maximum number of applications of the two lemmas. Note that in each step, we add at least one index to $I_1$ and then only go to the next step. Also, once an index is added to $I_1$ it is never removed. Thus, the total number of steps is at most $k$. Also, in each step, everytime we apply a lemma, the size of $I_1\cup I_2$ gets increased, which can be at most $k$. Thus, in each step we will apply the lemmas at most $k$ times in total. Hence, the total number of applications of both lemmas is $O(k^2)$. 

Next, we prove that the algorithm terminates with the desired configuration. Fix a step. Note that if $I_2=\emptyset$ and we make correct choices, we can correctly apply a lemma and make progress. This is true, as the conditions in the two lemmas are complementary. Now, if $I_2\ne\emptyset$, then we have applied Lemma \ref{lem:nonuni_third} at least once. This implies for every $i \in [k] \setminus(I_1 \cup I_2)$ and $j \in I_1$, such that $B_j$ and $B^\star_i$ intersect, $r_i \leq 5 r_j$. Hence, the condition of Lemma \ref{lem:nonuni_second} is false for the current set $I_1$. Now, we do not change $I_1$ throughout a step once we apply Lemma \ref{lem:nonuni_third}, except at the last time, in which case we go to the next step emptying $I_2$. Thus, once $I_2\ne\emptyset$, throughout the step, it holds that for every $i \in [k] \setminus(I_1 \cup I_2)$ and $j \in I_1$, such that $B_j$ and $B^\star_i$ intersect, $r_i \leq 5 r_j$. Hence, we can always apply Lemma \ref{lem:nonuni_third} and make progress. 

Now, consider the last step, we prove that at the end of this step $I_1\cup I_2=[k]$. By the above discussion, this step ends either if $I_1\cup I_2$ becomes $[k]$ or an index is added to $I_1$. In the latter case, we go to the next step. However, as the current step is the last one, it must be the case that $I_1\cup I_2=[k]$. This completes the proof of the claim, as we always maintain a valid configuration. 
\end{proof}

If the algorithm made all the correct random choices, by the above claim together with Lemma \ref{lem:nonuni_approx}, we can get a $15$-approximate solution from $\mathcal{B}_1$ and $\mathcal{B}_2$. The algorithm runs in linear time and succeeds with probability at least $\frac{\epsilon^k}{k^k\cdot n^2} \cdot \frac{1}{k^{2k}} \cdot \frac{1}{(6(k+1)^2)^{k^2}}$. This means that running the previous algorithm $ 2^{\cO(k^2\log k)}\cdot n^2$ times, we obtain a $(15+\epsilon)$-approximation algorithm with constant probability. Lastly, it is not hard to derandomize this algorithm by performing exhaustive searches in each step instead of making decisions randomly. The time bound still remains the same. 
\end{proof}

\subsection{{\radii} with Uniform Capacities in General Metrics}
In this subsection we show how to get a better approximation in the case where the capacities are uniform. 

\uniform*

\vspace{2mm}

\begin{tcolorbox}[colback=white!5,colframe=white!75!black]
From now on, let $\mathcal{B}^\star := \{B^\star_1, \cdots, B^\star_k \}$ denote the set of balls of a hypothetical optimal solution, $\mu^\star : P \rightarrow \mathcal{B}^\star$ be the associated assignment and for all $i \in [k]$, let $r^\star_i$ denote the radius of the ball $B^\star_i$. Let $\OPT$ be the optimal cost. Throughout this section, $n$ will denote the size  of $P$.
\end{tcolorbox}
\vspace{2mm}

As in the non uniform case, Lemma \ref{lem:guess_radius_general} implies that we can assume that the algorithm knows an approximate value $r_i$ of $r^\star_i$ for all $i \in [k]$.
Let $\mathcal{B}^\star_1$ denote the set of balls $B^\star_i$ of $\mathcal{B}^\star$ such that $(\mu^\star)^{-1}(B^\star_i) \geq U/k$ and $\mathcal{B}^\star_2 =  \mathcal{B}^\star \setminus \mathcal{B}^\star_1$. That is, $\mathcal{B}^\star_1$ are those balls which are assigned ``close to its capacity'', hence heavy, and remaining are in some sense light balls.  Since, there is at most $2^k$ ways of partitioning $\mathcal{B}^\star$ into $(\mathcal{B}^\star_1, \mathcal{B}^\star_2)$, by adding a factor $2^k$ to the running time, we can assume that the algorithm knows this partition. In the remaining of the section, it will be implicit that the partition of $\mathcal{B}^\star$ into $(\mathcal{B}^\star_1, \mathcal{B}^\star_2)$ is know. The following easy remark will be essential to our approach.

\begin{rmk}\label{rk:sampling_for_big}
	Let $B$ be a ball containing at least $U/k$ points of $P$. A point uniformly sampled from $P$ has probability at least $1/k^2$ to belong to $B$. 
\end{rmk}

\begin{proof}
	Indeed, since $B$ contains at least $U/k$ points and we assumed that $n \leq kU$, $B$ contains more than $n/k^2$ points of $P$, which completes the proof. 
\end{proof}

In particular, Remark~\ref{rk:sampling_for_big} implies that by doing $k$ uniform sampling,  we obtain with probability at least $1/k^{2k}$ a point in each ball of $\mathcal{B}^\star_1$. Further, because of Lemma \ref{lem:guess_radius_general} we know the radius, say $r_j$, of each ball and thus taking a ball of radius $2r_j$ around these sampled points, 
would cover everything $\mathcal{B}^\star_1$ would cover. In particular, if  $\mathcal{B}^\star_1 = \mathcal{B}^\star$, then this immediately gives a $(2+ \epsilon)$-approximation running in time 
$\cO(k^{2k}\cdot n^{3})$.

What remains to understand now are the balls of $\mathcal{B}^\star \setminus \mathcal{B}^\star_1 = \mathcal{B}^\star_2$. While we might not be able to sample inside these balls, the fact that they are assigned less than $U/k$ points each, allows us to take the union of these balls without breaking the capacity constraints. For example, replacing all the balls in $\mathcal{B}^\star_2$ with {\em one ball} containing every point in $P$ is still a valid solution to the problem, even though the increase in radius might be too big to give a good approximation. The idea will be to do this strategy only for some carefully selected subsets of $P$. 

For this purpose, we define a graph $G^\star$ where the vertices $v_i$ correspond to the balls $B^\star_i$ of $\mathcal{B}^\star$ and two vertices $v_i$ and $v_j$ are adjacent if the corresponding balls $B^\star_i$ and $B^\star_j$ intersect. For our purpose, it is sufficient that the algorithm knows the  vertices in each connected component of $G^\star$. Since, the number of connected components is at most $k$, by multiplying  $k^k$ to the running time, we can assume that we know  the connected components of $G^\star$.  

{\bf Therefore, for cleaner presentation we will assume that the connected components of $G^\star$, as well as the bipartition $(\mathcal{B}^\star_1, \mathcal{B}^\star_2)$ of $\mathcal{B}^\star$ are known in the rest of the section}.

Abusing notations, we will identify vertices of $G^\star$ as balls of $\mathcal{B}^\star$. For a connected component $C$, the \textit{radius} of $C$, denoted by $\rad(C)$, is defined as the sum of the $r_i$ for all the balls $B^{\star}_i$ in $C$. The next lemma can easily be proved by an induction on the number of balls inside $C$.

\begin{lemma}\label{lem:component_and_radius}
	If $C$ is a connected component of $G^\star$, then the set of points $X$ contained inside the balls of $C$ is such that the maximum distance between any two points of $X$ is at most $2\cdot \rad(C)$.
\end{lemma}

In particular,  Lemma~\ref{lem:component_and_radius} implies that if $G^\star$ is connected, then removing all the balls in $\mathcal{B}^\star_2$ and taking {\em one ball} containing all the points of $P$ only adds $2 \cdot \OPT$
to the cost. However, if $G^\star$ is not connected, then we can do this procedure for every component $C$ containing a ball of $\mathcal{B}^\star_1$, as we are able to {\em sample a point in them} and thus the ball of radius $2\cdot (\rad(C))$ around this sampled point cover all the points in $C$ and, in particular, contains all the balls of $\mathcal{B}^\star_2 \cap C$. By doing this for all components of $C$ containing elements of both $ \mathcal{B}^\star_1 $ and $ \mathcal{B}^\star_2 $, we are able to approximate all the balls of $\mathcal{B}^\star_2$ contained in those components by paying an extra $2 \cdot \OPT$. What remains to approximate are the balls in $\mathcal{B}^\star_2$ that are not in a connected component of $C$ containing an element of $\mathcal{B}^\star_1$. Let this set of balls be called  $\mathcal{B}^\star_3$. 

Let us start by precisely describing the way the algorithm guesses approximate balls for $\mathcal{B}^\star\setminus \mathcal{B}^\star_3$. We will say that a list $L_1$ of balls \textit{approximates} $\mathcal{B}^\star_1$ if for every ball $B^\star_i \in \mathcal{B}^\star_1$, $L_1$ contains a ball $B_i$ such that $B^\star_i \subseteq B_i$ and $\rad(B_i) = 2 \cdot r_i$. We will also say that $B_i$ \textit{approximates} $B^\star_i$. Moreover, a list $L_2$ of balls is said to \textit{cover} $G^\star$ if for every connected component $C$ of $G^\star$ containing a ball of $ \mathcal{B}^\star_1$ and a ball of $ \mathcal{B}^\star_2$, there exists a ball $B_C \in L_2$ of radius $2 \cdot \rad(C)$ that contains all the points of $C$.

\begin{lemma}[{\bf Construction of $L_1\cup L_2$}]
\label{lem:general_first step}
	There exists an algorithm running in linear time that finds, with probability at least $\frac{1}{k^{2k}}$ a list $L_1$ approximating $\mathcal{B}^\star_1$ and a list $L_2$ covering $G^\star$. 
\end{lemma}

\begin{figure}[!ht]
		\begin{center}
			\includegraphics[width=.7\textwidth]{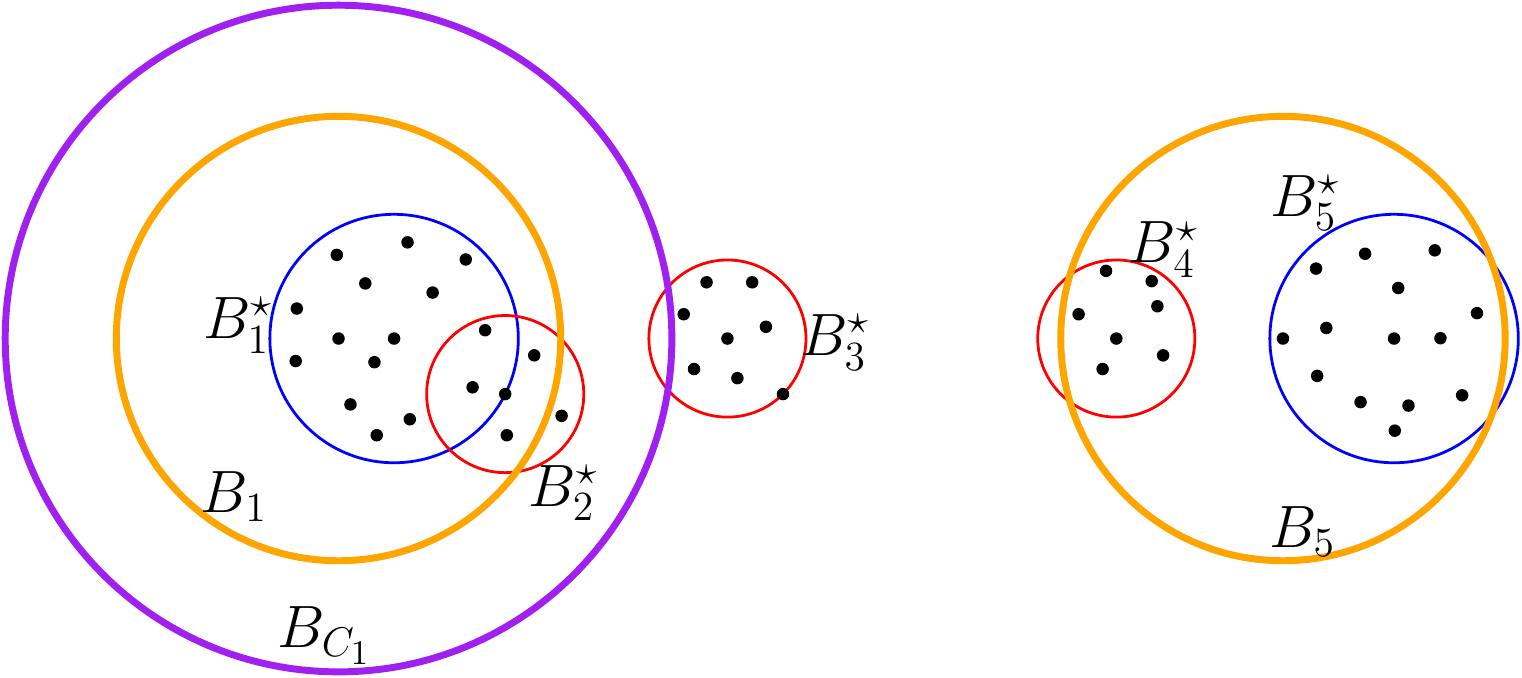}
			\vspace{-2mm}
			\caption{Figure showing construction of $L_1\cup L_2$. $G^\star$ has five components: $C_1=\{B^\star_1,B^\star_2\}$, $C_2=\{B^\star_3\}$, $C_3=\{B^\star_4\}$, $C_4=\{B^\star_5\}$. Also, $\mathcal{B}^\star_1=\{B^\star_1,B^\star_5\}$ and $\mathcal{B}^\star_2=\{B^\star_2,B^\star_3,B^\star_4\}$. For this example, $L_1=\{B_1,B_5\}$ and $L_2=\{B_{C_1}\}$.}
			\label{fig:fig1}
		\end{center}
\end{figure}

\begin{proof}
	This proof follows   from Remark \ref{rk:sampling_for_big}. For every $i$ such that $B^\star_i$ belongs to $\mathcal{B}^\star_1$, we sample uniformly at random a point $x$ in $P$ and take for $B_i$ the ball of radius $2r_i$ around $x$. Moreover, if $i$ is the smallest index such that $B^\star_i$ belongs to some component $C$ containing elements of $\mathcal{B}^\star_1$ and $\mathcal{B}^\star_2$, then we also take the ball $B_C$ of radius $2\cdot \rad(C)$ around $x$. (See Figure \ref{fig:fig1} for an example.) The algorithm succeeds if every sample is chosen from a distinct ball in $\mathcal{B}^\star_1$, so  in total with probability at least $\frac{1}{k^{2k}}$ the algorithm succeeds. 
\end{proof}

From now on, we assume that the algorithm from Lemma \ref{lem:general_first step} succeeded and we have a list $L_1$ approximating $\mathcal{B}^\star_1$ and a list $L_2$ covering $G^\star$. We cover the rest of the points not covered by balls in the lists $L_1$ and $L_2$ by a greedy algorithm that is described below. 
 \begin{description}
\item[Construction of $L_3$.]  At this stage, what remains is to select ``approximate  balls'' corresponding to balls in $\mathcal{B}^\star_3$. If there is a point of $P$ which is not covered in the union of the balls in $L_1$ and $L_2$, then it has to be covered by a ball in $\mathcal{B}^\star_3$ in the optimal solution.  When we say {\em not covered}, then we mean a point that is not in the union of balls in  $L_1$ and $L_2$. Further, note that all the points that are covered by balls in $L_1$ and $L_2$ {\em may not be assigned to these balls} because of the capacity constraint. 

So, for the greedy algorithm, we do as follows. Select a point not covered by already selected balls, 
guess which ball it belongs to (so in particular the radius $r_i$ of it) and take the ball of radius $2r_i$ around that point. We do this greedily until all points of $P$ are covered, we get a list $L_3$ of balls such that all the points of $P$ are covered. (See Figure \ref{fig:fig2} for an example.)
\end{description}

Formally, if $L_1$ and $L_2$ are two lists of balls such that $L_1$ approximates $\mathcal{B}^\star_1$ and $L_2$ covers $G^\star$, a list $L_3$ \textit{complements} $L_1$ and $L_2$ if the balls of $L_1, L_2$ and $L_3$ covers all points of $P$ and every ball $B_j$ of $L_3$ is associated to a unique ball $B^\star_j$ in $\mathcal{B}^\star_3$ such that $B^\star_j \subseteq B_j$, $\rad(B_j) = 2\cdot  r_j$.

\begin{lemma}\label{lem:l3}
	Let $L_1, L_2$ be two list of balls such that $L_1$ approximates $\mathcal{B}^\star_1$ and $L_2$ covers $G^\star$. There exist a polynomial time algorithm that finds, with probability at least $\frac{1}{k^k}$, a list $L_3$ complementing $L_1$ and $L_2$.
\end{lemma}

\begin{proof}
	The algorithm starts with $L_3 := \emptyset$ and as long as there exists a point $x$ not covered by $L_1 \cup L_2 \cup L_3$, picks uniformly at random an index $i$ and add the ball $B_i$ of center $x$ and radius $2r_i$ to $L_3$. 

	The algorithm outputs the desired set if at every step, the index $i$ selected and the uncovered point $x$ are such that $\mu^\star(x) = B^\star_i$. Overall since there is at most $k$ steps, we have that the probability of success is at least $\frac{1}{k^k}$.
\end{proof}

Assume now that we have the sets $L_1, L_2$ and $L_3$ as described in Lemma \ref{lem:l3}. Note that, while the union of the balls in these sets covers $P$, this does not imply that we can find an assignment $\mu$ of points to balls satisfying the capacity constraints. Thus, we might have to select more balls in order to not only cover all the points but also be able to find an assignment satisfying the capacity constraint.  Further, note that if points in $P$ cannot be assigned to balls in $L_1\cup L_2\cup L_3$  without satisfying capacity constraints, then there exists a ball in $\mathcal{B}^\star_3$ corresponding to which an approximate ball has not been selected in 
$L_1\cup L_2\cup L_3$. Formally, let $L_1, L_2$ and $L_3$ be sets of balls such that $L_1$ approximates $\mathcal{B}^\star_1$, $L_2$ covers $G^\star$ and $L_3$ complements $L_1$ and $L_2$, we say that a ball $B^\star_j \in \mathcal{B}^\star_3$ is \textit{untouched} by $L_1, L_2$ and $L_3$ if no ball $B_i$ of $L_3$ is associated to $B^\star_j$. Let $I(L_1, L_2, L_3)$ denote the set of indices of balls untouched by $L_3$.

 \begin{description}
 \item[Increase in Radius:]
 Let $i \in I(L_1, L_2, L_3)$. Remember that $r_i$ is the approximated radius of the ball  $B^\star_i$ and the points of $B^\star_i$ must {\em entirely be covered} in the balls of $L_1 \cup L_2 \cup L_3$. Suppose first that there exists a ball $B$ of $L_2$ or $L_3$ intersecting $B^\star_i$. If this is the case, then we will increase the radius of the said ball $B$ by $2 \cdot r_i$ so that $B^\star_i$ is now entirely contained in $B$. Moreover, since the balls in $L_2$ and $L_3$ will only be assigned points originally assigned to $\mathcal{B}^\star_2$ by $\mu^\star$, it is possible to assign the points of $ (\mu^\star)^{-1}(B^\star_i)$ to this ball as well. 
We will prove it formally in Lemma~\ref{lem:good_approx_general}. 

 \end{description}
Formally, if $L_1, L_2$ and $L_3$ are lists of balls such that $L_1$ approximates $\mathcal{B}^\star_1$, $L_2$ covers $G^\star$ and $L_3$ complements $L_1$ and $L_2$, then we denote $ I_1(L_1, L_2, L_3)$ the set of of indices $i \in I(L_1, L_2, L_3)$ such that $B^{\star}_i$ intersects a ball $B$ of $L_2 \cup L_3$. We also let $ I_2(L_1, L_2, L_3) = I(L_1, L_2, L_3) \setminus I_1(L_1, L_2, L_3)$, and for an index $i \in I_1(L_1, L_2, L_3)$ by $B[L_1, L_2, L_3](i)$ some arbitrary ball of $L_2 \cup L_3$ intersecting $B^\star_i$.

\begin{description}

\item[Construction of $L_4$.] What remains now is to consider the case where $i \in I_2(L_1, L_2, L_3)$, which means that $B^\star_i \in\mathcal{B}^\star_3 $ is covered by balls in $L_1$ only. Suppose first that there is a ball $B_j$ in $L_1$ containing some points of $B^\star_i$ and such that the connected component $C$ containing $B^\star_j$ does not contain any ball in $\mathcal{B}^\star_2$. Recall that $B_j$ is {\em larger than} $B^\star_j$, so it is possible for $B_j$ to intersect $B^\star_i$ but $B^\star_j$ need not intersect $B^\star_i$. 
In order to {\em not count} $r_i$ multiple times, if there are multiple choices for $B_j$, we choose one as follows. If there exists a $B_j$, such that  the component in $G^\star$, containing the  corresponding $B^\star_{j}$, does not contain any ball of $\mathcal{B}^\star_2$, then we assign $B^\star_i$ to this $B_j$,  and say that $B^\star_i$ is \textit{associated} with  $B_j$. If there are several choices for such $B_j$, we choose one arbitrarily. Further, if {\em no such} $B_j$ exists, then we do not {\em assign any ball} to this $B^\star_i$. 

For every $B_j$ in $L_1$ such that $B^\star_{j_1}, \dots, B^\star_{j_w}$ denote the set of balls in $\mathcal{B}^\star_3$ associated to $B_j$, we define $T_j$ to be the ball centered at the center of $B_j$ and of radius $2\cdot (\rad(B^\star_j) + \sum_{i \in [w]} \rad(B^\star_{j_i}))$. Note that all the $B^\star_{j_i}$ are contained in $T_j$, as they intersect with $B_j$, and hence we will be able to assign these balls to $T_j$.
 Let $L_4$ be the list of all the $T_j$ for all the balls in $L_1$ such that some balls of $\mathcal{B}^\star_3$ are associated with them. (See Figure \ref{fig:fig2} for an example.)
\end{description}

\begin{figure}[!t]
		\begin{center}
			\includegraphics[width=.7\textwidth]{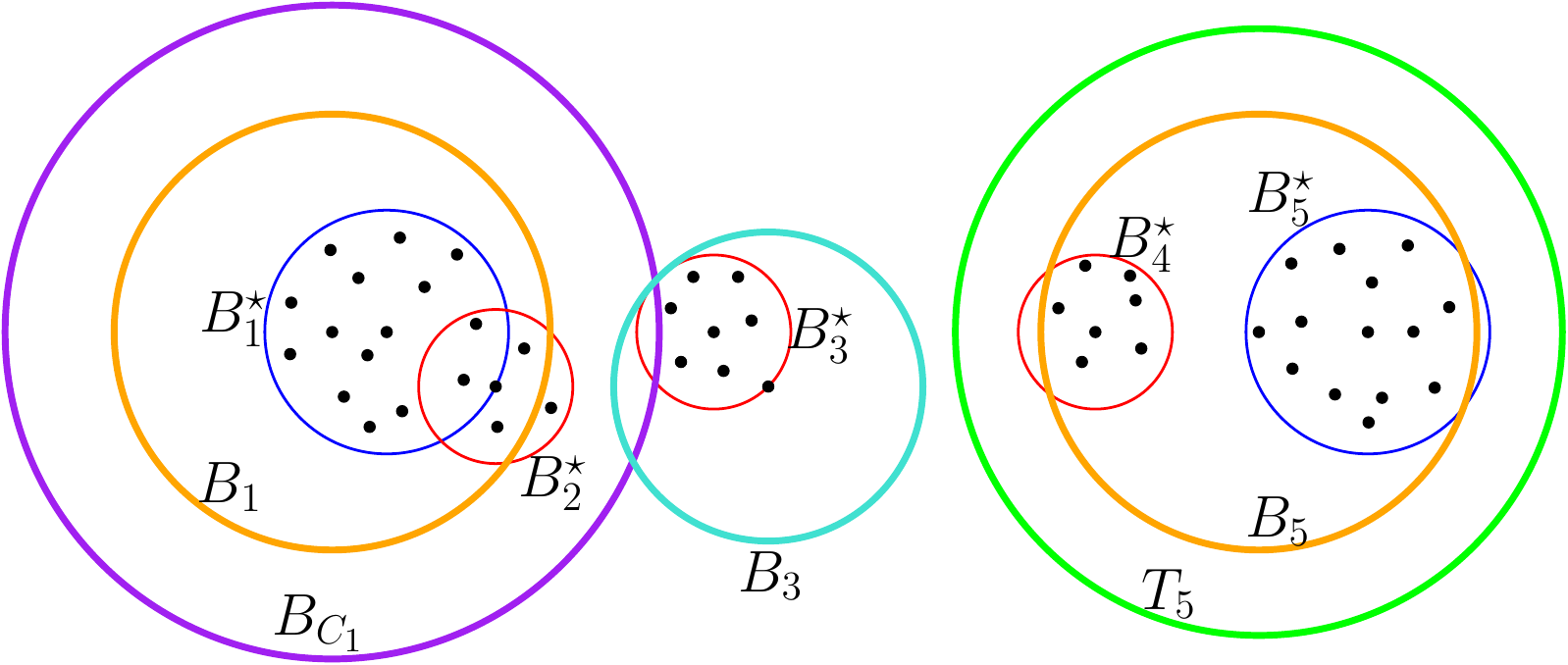}
			\vspace{-2mm}
			\caption{Figure showing construction of $L_3$ and $L_4$ for the example in Figure \ref{fig:fig1}. Here  $L_3=\{B_3\}$ and $L_4=\{T_5\}$.}
			\label{fig:fig2}
		\end{center}
\end{figure}

Observe that for some balls in $B^\star_i \in\mathcal{B}^\star_3 $, we did not associate any balls, and hence did not include a ball corresponding to them in $L_4$. Let us understand what kind of balls these are, why we did not associate any ball in $L_4$, and how we take care of them. Let  $B^\star_i \in \mathcal{B}^\star_3$  
be covered by some balls $B_{i_1}, B_{i_2}, \dots, B_{i_s}$ of $L_1$ for some $i_s \in [k]$. Moreover, each of the $B_{i_j}$ is such that the corresponding $B^\star_{i_j}$ belongs to some component $C_j$ containing some ball of $\mathcal{B}^\star_2$. Therefore, $L_2$ contains a ball $B_{C_j}$ and we already pay $2\cdot \rad(B^\star_{i_j})$ for the radius of $B_{C_j}$. This is {\em the reason}  we cannot take a ball of radius $2\cdot (\rad(B^\star_{i_j}) + \rad(B^\star_i))$ to approximate $B^\star_i$ as we just did with $L_4$, as this will imply that the solution would cost at least $6 \cdot \rad(B^\star_{i_j})$.
This is the reason, we did not add any new ball.   Let us remark that this strategy will indeed give a factor $6+\epsilon$ approximation for our problem.  We will show how we can ``cover'' the points covered by $B^\star_i$ by already selected balls. 
So for every $j \in [s]$, we will assign all the points of $X_j = B^\star_i \cap B_{i_j}$ to $B_{i_j}$. So the points covered by $B^\star_i$ is reassigned to balls in $B_{i_1}, B_{i_2}, \dots, B_{i_s}$. 
This is always possible because $ B^\star_{i_j} \subseteq B_{C_j}$. However,  by doing this we increase the number of points assigned to $B_{C_j}$ but it still remains smaller than $\frac{U}{k} \cdot | \mathcal{B}^\star_2| \leq U$. It means that we can cover as well as assign the points covered by  $B^\star_i$ using balls in $L_1 \cup L_2 \cup L_3$. 

Formally, if $L_1, L_2$ and $L_3$ are lists of balls such that $L_1$ approximates $\mathcal{B}^\star_1$, $L_2$ covers $G^\star$ and $L_3$ complements $L_1$ and $L_2$, then $I_3(L_1, L_2, L_3)$ denote the set of indices $i \in I_2(L_1, L_2, L_3)$ such that there exists an index $j \in [k]$ with $B^{\star}_i \cap B_j \not = \emptyset$ and $B^{\star}_j$ is in a component of $G^\star$ without any element of $\mathcal{B}_2$. Moreover, for every $i \in I_3(L_1, L_2, L_3)$, we let $B'[L_1, L_2, L_3](i)$ denote that index $j$ (in case of multiple choices, we pick an arbitrary one). For any $j \in [k]$ such that $B^{\star}_j \in \mathcal{B}_1$, we let $A(L_1, L_2, L_3)(j)$ denote the set of indices in $I_3(L_1, L_2, L_3)(j)$ such that $B'[L_1, L_2, L_3](i) = j$.

The next two lemmas formalize the previous discussion. First let us show that the set of balls defined until now indeed form a valid solution.

\begin{lemma}\label{lem:good_approx_general}
	Let $L_1, L_2$ and $L_3$ be list of balls such that $L_1$ approximates $\mathcal{B}^\star_1$, $L_2$ covers $G^\star$ and $L_3$ complements $L_1$ and $L_2$.  
	\begin{itemize}
		\item Let $L'_2$ be the set that contains for every $B \in L_2$, if we note $i$ the index in $I_1(L_1, L_2, L_3)$ with $B[L_1, L_2, L_3](i) = B$ that maximises $r_i$, the ball $\ext(B, 2r_i)$ if this index exists. If this index doesn't exists, then $L'_2$ contains $B$.
		\item Let $L'_3$ be the set that contains for every $B \in L_3$, if we note $i$ the index in $I_1(L_1, L_2, L_3)$ with $B[L_1, L_2, L_3](i) = B$ that maximises $r_i$, the ball $\ext(B, 2r_i)$ if this index exists. If this index doesn't exists, then $L'_3$ contains $B$.
		\item Let $L_4$ be the set obtained by taking, for every index $j \in [k]$ such that $B^{\star}_j \in \mathcal{B}_1$ and the connected component $C$ containing $B^\star_j$ does not contain any ball in $\mathcal{B}^\star_2$, the ball $T_j$ centered at the center of $B_j$ the ball approximating $B^{\star}_j$ in $L_1$ and of radius $2\cdot (\rad(B^\star_j) + \sum_{i \in A(L_1, L_2, L_3)(j) } \rad(B^\star_{j_i}))$.
	\end{itemize}

	The set $ \mathcal{L} = L_1 \cup L'_2 \cup L'_3 \cup L_4$ is a valid solution to the \radii problem defined by $(P,\dist)$ of value $4 \cdot \sum_{i\in [k]} r_i$

\end{lemma}

\begin{proof}
	Note that, by construction, $L'_2$ (resp. $L'_3$) contains a copy of each ball of $L_2$  (reps. $L_3$) with the additional property that some of these balls get extended.  
	By definition of complementing, the balls in $L_1 \cup L_2 \cup L_3$ cover $P$ and therefore so does $L_1 \cup L'_2 \cup L'_3 \cup L_4$. To show that this is indeed a valid solution to the problem, we will define a valid assignment $\mu: P \rightarrow \mathcal{L}$. Recall, that $\mathcal{B}^\star := \{B^\star_1, \cdots, B^\star_k \}$ denotes  the set of balls of a hypothetical optimal solution and $\mu^\star : P \rightarrow \mathcal{B}^\star$ is the associated assignment.

	\begin{itemize}
	\setlength{\itemsep}{-2pt}
		\item Let $x$ be a point of $P$ with $\mu^\star(x) = B^\star_i$ and $B^\star_i$ belongs to $\mathcal{B}^\star_1$. In that case, we know that there exists a ball $B_i \in L_1$ such that $B^\star_i \subseteq B_i$, and we let $\mu(x) = B_i$. 
		\item Let $x$ be a point of $P$ with $\mu^\star(x) = B^\star_i$ where $B^\star_i$ is a ball of $\mathcal{B}^\star_2$ such that, in $G^\star$, the connected component $C$ containing $B^\star_i$ contains a ball of $\mathcal{B}^\star_1$. In that case, $L'_2$ contains a ball $B_C$ containing all the points in a ball of $C$ and we let $\mu(x) = B_C$.
		\item Let $x$ be a point of $P$ with $\mu^\star(x) = B^\star_i$ where $B^\star_i$ is a ball of $\mathcal{B}^\star_2$ entirely contained in some ball $B \in L'_2 \cup L'_3$. We set $\mu(x) = B$.
		\item Let $x$ be a point of $P$ such that $\mu^\star(x) = B^\star_i$ for some ball $B^\star_i$ of $\mathcal{B}^\star_2$ where there exists a ball $B_j$ of $L_4$ containing it. In that case, we let $\mu(x) = B_j$.
	\end{itemize}
	
\noindent 	
	Until now, no ball has been assigned more than $U$ points. Indeed, balls in $L_1$ get assigned exactly the same set of points as their corresponding ball in $\mathcal{B}^\star_1$ and balls in $L'_2 \cup L'_3 \cup L_4$ only gets assigned points that were assigned by $\mu^\star$ to balls in $\mathcal{B}^\star_2$. 
	
	The points for which $\mu(x)$ has not been defined so far are the points covered by some ball in $B^\star_i \in \mathcal{B}^\star_2$ where $B^\star_i$ itself is covered by some balls $B_{i_1}, B_{i_2}, \dots, B_{i_s}$ of $L_1$.
%
%
%
	Moreover each of the $B_{i_j}$ are such that the corresponding $B^\star_{i_j}$ belongs to some component $C_j$ containing some ball of $\mathcal{B}^\star_2$ which means that $L_2$, and thus $L'_2$, contains a ball $B_{C_j} $ covering all points of $C_j$. Therefore, for any such ball $B^\star_i$ and ball $B_{i_t} \in L_1$ corresponding to some $B^\star_{i_t}$ in $\mathcal{B}^\star_1$ we let $\mu(x) = B_{i_t}$ for every $x \in  B_{i_t} \cap B^\star_i$.  
Doing this, we have defined $\mu$ for all $P$, however {\em some balls} of $L_1$ might have been assigned more than $U$ points, we will show how we can {\em reassign} these points so that no capacity constraint is violated. To do so, we will make use of the ball $B_{C_j} \in L'_2$, the ball corresponding to the component $C_j$. 
Next we explain the reassignment procedure.  

Let $B_i$ be a ball violating the capacity constraint, corresponding to $B^\star_i$ in $\mathcal{B}^\star_1$.  By construction there exists a ball $B_C$ in $L'_2$ containing all the points of the component $C$ containing $B^\star_i$. In particular, $B_C$ contains $B^\star_i$. Let $B^\star_{i_1}, \dots, B^\star_{i_t}$ denote the set of balls in $\mathcal{B}^\star_2$ such that the points in $ B^\star_{i_t} \cap B_i$ have been assigned to $B_i$ in the previous step. Note that, at this moment, 
$$|\mu^{-1}(B_i)| = |(\mu^\star)^{-1}(B^\star_i)| + \sum_{j \in t}  |B^\star_{i_j} \cap B_i|.$$ 
Let $X$ be a set of $(\mu^\star)^{-1}(B_i)$ of size the minimum of $|(\mu^\star)^{-1}(B^\star_i)|$ and $\sum_{j \in t}  |(\mu^\star)^{-1}(B^\star_{i_j}) \cap B_i|$, and set $\mu(X) = B_C$. After doing this for all balls $B_i$ of $L_1$ that were exceeding capacity, it is easy to see that $\mu^{-1}(B_i)$ has size at most $U$. Moreover, while we increased the number of points assigned to some $B_C \in L_2$, it does not exceed 
$|\mathcal{B}^\star_2| \cdot \frac{U}{k} \leq k$, which concludes the proof that $\mu$ is a valid assignment.

Finally we give an upper bound on the value of $ \mathcal{L} $ as a solution. To achieve this, we partition $[k]$ into:

\begin{itemize}
	\item $H_1$ -- the set of of indices $i$ such that the ball $B^\star_i$ is in $\mathcal{B}^\star_1$ and is contained in a component $C$ of $G^\star$ where $C$ contains a ball of $\mathcal{B}^\star_2$;
	\item $H_2$ -- the set of of indices $i$ such that the ball $B^\star_i$ is in $\mathcal{B}^\star_1$ and is contained in a component $C$ of $G^\star$ where $C$ contains no ball of $\mathcal{B}^\star_2$;
	\item $F_1$ -- the set of indices $i$ such that there exists a ball $B_i$ in $L_3$ associated to $B^{\star}_i$. 
	\item $F_2$ -- the set of of indices $i$ such that the ball $B^\star_i$ is in $\mathcal{B}^\star_2$ and is contained in a component $C$ of $G^\star$ where $C$ contains a ball of $\mathcal{B}^\star_1$;
	\item $I_1(L_1, L_2, L_3)$
	\item $I_3(L_1, L_2, L_3)$ 
	\item $F_5$ the rest. 
\end{itemize}

	\begin{claim}
		The following are true:
		\begin{enumerate}
			\item $\rad(L_1) \leq 2 \cdot (\sum_{i \in H_1} r_i + \sum_{i \in H_2} r_i)$
			\item $\rad(L'_2) + rad(L'_3) \leq 2 \cdot ( \sum_{i \in H_1} r_i + \sum_{i \in F_2} r_i +  \sum_{i \in F_1} r_i + \sum_{i \in I_1(L_1, L_2, L_3)} r_i)$
			\item $\rad(L_4) \leq 2 \cdot (\sum_{i \in H_2} r_i + \sum_{i \in I_3(L_1, L_2, L_3)} r_i)$
		\end{enumerate}
	\end{claim}
	
	\begin{proof}
		The fact that $\rad(L_1) \leq 2 \cdot (\sum_{i \in H_1} r_i + \sum_{i \in H_2} r_i)$ follows from the definition of $L_1$. Moreover, by definition of $L_2$, we have that $\rad(L_2) \leq 2 \cdot (\sum_{i \in H_1} r_i + \sum_{i \in F_2} )r_i$. Likewise by definition of $L_3$ we have that $\rad(L_3) \leq 2 \cdot \sum_{i \in F_1} r_i$. Moreover, the difference between $L_2 \cup L_3$ and $L'_2 \cup L'_3$ is that some balls get extended by $2r_i$ for some $i \in I_1(L_1, L_2, L_3)$. Since each element of $I_1(L_1, L_2, L_3)$ contributes to the increase of at most one ball, we obtain indeed that $\rad(L'_2) + rad(L'_3) \leq 2 \cdot ( \sum_{i \in H_1} r_i + \sum_{i \in F_2} r_i +  \sum_{i \in F_1} r_i + \sum_{i \in I_1(L_1, L_2, L_3)} r_i)$.
		Finally, the fact $\rad(L_4) \leq 2 \cdot (\sum_{i \in H_2} r_i + \sum_{i \in I_3(L_1, L_2, L_3)} r_i)$ directly follows from the definition of $L_4$ and the fact that $B'[L_1, L_2, L_3](i)$ is unique for every $i \in I_3(L_1, L_2, L_3)$. 
	\end{proof}
	Using the previous claim we can verify that $\rad(\mathcal{L}) \leq 4 \cdot \sum_{i \in [k]} r_i$ , which ends the proof. 
\end{proof}

Finally, we are ready to prove Theorem~\ref{thm:general} 

\begin{proof}[Proof of Theorem~\ref{thm:general}]

	The algorithm starts by applying the algorithm Lemma \ref{lem:guess_radius_general} to obtain some values $r_1, \dots, r_k$. Then the algorithm will uniformly at random partition $[k]$ into two sets $(I_1,I_2)$. With probability at least $1/2^k$, $I_1$ is exactly the set of indices $i$ such that $B^{\star}_i \in \mathcal{B}^{\star}_1$. Then the algorithm will uniformly at random partition $[k]$ into $k$ subsets $C_1, \dots, C_k$. With probability at least $1/k^k$ every set $C_j$ corresponds to the set of indices of the balls inside the $j$-th connected component of $G^{\star}$. 
	
	Assuming it knows the partition $(\mathcal{B}^{\star}_1, \mathcal{B}^{\star}_2)$, the connected components of $\mathcal{G}^{\star}$, the algorithm will then run the algorithm of Lemma \ref{lem:general_first step} to obtain two sets $L_1$ and $L_2$. Assuming the run of the algorithm of Lemma \ref{lem:guess_radius_general} was succesfull, with probability at least $\frac{1}{k^{2k}}$, $L_1$ approximates $\mathcal{B}^\star_1$ and $L_2$ covers $G^\star$. Then the algorithm will run the algorithm of Lemma \ref{lem:l3} and obtain a list $L_3$ which, with probability at least $\frac{1}{k^k}$, complements $L_1$ and $L_2$. Finally, the algorithm then applies Lemma \ref{lem:good_approx_general} to obtain a valid solution of value $4 \cdot \sum_{i\in [k]} r_i$. To do so, assuming $L_1$ approximates $\mathcal{B}^\star_1$, $L_2$ covers $G^\star$ and $L_3$ complements $L_1$ and $L_2$, the algorithm only needs to know: 
	
	\begin{itemize}
		\item the set $I_1(L_1, L_2, L_3)$;
		\item the set $I_3(L_1, L_2, L_3)$;
		\item the value of $B[L_1, L_2, L_3](i)$ for each $i \in I_1(L_1, L_2, L_3)$; and
		\item the set $A(L_1,L_2,L_3)(j)$ for all index $j$ such that the ball $B^\star_j$ is in $\mathcal{B}^\star_1$ and is contained in a component $C$ of $G^\star$ where $C$ contains no ball of $\mathcal{B}^\star_2$. 
	\end{itemize}

	To do so, the algorithm will pick uniformly at random a value for all these sets and the probability of success is at least $1/k^{\cO(k)}$. 

	Overall, the algorithm outputs a valid solution of value $4 \cdot \sum_{i\in [k]} r_i$ if all its uniform sampling and call to the algorithms of Lemmas \ref{lem:guess_radius_general}, \ref{lem:general_first step} and \ref{lem:l3} are successful. All this can be done in linear time and the probability that all these events are true is at least $\frac{1}{n^2 \cdot (k)^{\cO(k)}} $. This means that, running the algorithm  $(\frac{k}{\epsilon})^{\cO(k)}n^2$  and outputting the valid solution with smallest value will give us a valid solution of value $4 \cdot \sum_{i\in [k]} r_i \leq 4 \cdot (1 + \epsilon) OPT$ with constant probability, which ends the proof.
\end{proof}
\subsection{The Bi-criteria Approximation}

The bi-criteria approximation is a by product of our random sampling scheme described in the previous subsection. We use the same approach of guessing the 2-approximate balls of $\mathcal{B}^\star_1$. However, we slightly modify the definition of $\mathcal{B}^\star_1$. Let $\mathcal{B}^\star_1$ denote the set of balls $B^\star_i$ of $\mathcal{B}^\star$ such that $(\mu^\star)^{-1}(B^\star_i) \geq \epsilon U/k$ and $\mathcal{B}^\star_2 =  \mathcal{B}^\star \setminus \mathcal{B}^\star_1$. Then we can trivially extend Lemma \ref{lem:general_first step} to obtain the following lemma. 

\begin{lemma}
\label{lem:general_first step bi}
	There exists a polynomial time algorithm running in linear time that finds, with probability at least $\frac{\epsilon^k}{k^{2k}}$, a list $L_1$ that contains for every ball $B^\star_i \in \mathcal{B}^\star_1$ a ball $B_i$ containing $B^\star_i$  such that $\rad(B_i) = 2 \cdot \rad(B^\star_i)$.
\end{lemma}

We apply the above lemma to compute the list $L_1$. Assume that $L_1$ has the desired property. Let $P'$ be the set of points not contained in the union of the balls of $L_1$. Note that the size of $P'$ can be at most $(\epsilon U/k)\cdot k=\epsilon U$, as the points in $P'$ can only be contained in balls of $\mathcal{B}^\star_2$. We apply the following procedure to select a set of balls to cover the points of $P'$. For the time being, assume that we know the index of the ball in $\mathcal{B}^\star_2$ that is assigned with $p\in P'$ in $\mu^\star$. 

Take any point $p\in P'$. Let $p$ be  assigned to $B_j\in \mathcal{B}^\star_2$ with radius $r_j$. Add the ball $B$ centered at $p$ of radius $2r_j$ to a list $L_2$. Remove the points of $B$ from $P'$. Repeat the process until $|\mathcal{B}^\star_2|$ balls are added to $L_2$ or $P'$ is exhausted. 

Note that as per our assumption, that we know the optimal assignment of the points in $P'$, the above procedure finds at most $|\mathcal{B}^\star_2|$ balls such that the union of the balls in $L_2$ contains all the points if $P'$. Also, it is easy to remove the assumption by enumerating all possible assignments of the points of $P'$ which are chosen as centers. The number of such assignments is at most $k^{\cO(k)}$. 

We can also prove the existence of a valid assignment of the points in $P$ to the balls in $L=L_1\cup L_2$ such that each ball is assigned at most $(1+\epsilon)U$ points. For each index $i$ of the balls in $\mathcal{B}^\star_1$, assign the points in $(\mu^\star)^{-1}(B^\star_i)$ to $B_i\in L_1$. Note that the number of points that are not yet assigned is at most $\sum_{B\in \mathcal{B}^\star_2} |(\mu^\star)^{-1}(B)| \le (\epsilon U/k)\cdot k=\epsilon U$. Assign each remaining point to an arbitrary ball of $L$ that contains it. 

By repeating Lemma \ref{lem:general_first step bi} $\cO((k^2/\epsilon)^k)$ times, we can ensure that $L_1$ has the desired property. The procedure of covering the points of $P'$ runs in $k^{\cO(k)} \cdot n^{\cO(1)}$ time. Thus, in total our algorithm takes $(k^2/\epsilon)^{\cO(k)}\cdot n^{\cO(1)}$ time. 

\genbi*
\section{{\radii}: Euclidean Metric}
\label{sec:Euclidean}
In this section, we consider the uniform-capacitated version. Thus, all the capacities are $U$. Let $B^\star_1, \dots, B^\star_k$ denote the optimal solution and $r^\star_i$ the radius of $B^\star_i$ for every $i \in [k]$. Without loss of generality, we can assume that $r^\star_1 \geq r^\star_2, \dots, \ge r^\star_k$. As in the general case, we can guess an approximation of the radii. 

\begin{lemma}\label{lem:guess_Rd}
	There exists a randomized algorithm running in $O(nd)$ time and with probability at least $\frac{\epsilon^k}{(2k)^k\cdot n^2}$ outputs a set $r_1, \dots r_k$ of reals such that for every $i \in [k]$, we have $r_i^\star \leq r_i$ and $\sum_{i \in [k]} r_i \leq (1 + \epsilon)  \sum_{i \in [k]} r^\star_i$.
\end{lemma}

\begin{proof}
	The algorithm starts by picking uniformly at random a pair of points $x,y$ of $P$. With probability $1/n^2$, they correspond to the pair of points of $B^\star_1$ of maximum distance. In particular, it means that $r^\star_1 \leq 2 d(x,y)$. This means that guessing each $r^\star_i$ up to a precision of $\frac{\epsilon d(x,y)}{k}$ would yield the desired result. However, since $r^\star_i \leq r^\star_1 \leq 2 d(x,y)$, there are only $r^\star_i \cdot \frac{k}{\epsilon d(x,y)} \leq \frac{2k}{\epsilon} $ choices for each $r^\star_i$, which gives $(\frac{2k}{\epsilon})^k$ possible total choices. 
\end{proof}


Badoui~et al.~\cite{Badoui2002} used the following lemma to find a core-set for the $k$-center problem in $R^d$.

\begin{lemma}\label{lem:corest_next}
	Let $S$ be a set of points in $R^d$ and $B$ be the minimum enclosing ball of $S$. Noting, $c$ and $r$ the center and radius of $B$, we have that for any $ \epsilon >0$, if $p$ is a point at distance more than $(1 + \epsilon) r$ from $c$, then the minimum enclosing ball of $S \cup \{p\}$ has radius at least $(1 + \epsilon^2/16)r$.  
\end{lemma} 

In particular, it allows them to find a $(1 + \epsilon)$-approximation algorithm for $k$-center in $2^{\cO(k \log(k)/\epsilon^2)}$ time. Another interesting consequence of Lemma \ref{lem:corest_next} is the following lemma:

\begin{lemma}\label{lem:coreset}
	Let $S$ be a set of points in $R^d$ and $B$ be the minimum enclosing ball of $S$. There exists a set $S' \subseteq S$ of size $\cO(1/\epsilon^2)$ such that, noting $B'$ the minimum enclosing ball of $S'$, the $(1+ \epsilon)$ expansion of $B'$ contains $B$.
\end{lemma}

 Lemma \ref{lem:corest_next} implies that for every $i \in [k]$ there is a set $S'_i$ of size $\cO(1/\epsilon^2)$ such that the $(1 + \epsilon)$ expansion of the minimum enclosing ball of $S'_i$ contains $B_i$. Using this observation, one can find a trivial algorithm for \radii. This algorithm tries all possible values of the $k$ sets $S'_i$, which is at most $n^{\cO(k/\epsilon^2)}$.

A very natural question is to find an approximation algorithm for the problem in $R^d$ with better running time. In this section, we first obtain a $(2 +\epsilon)$-approximation running in time $f(k,\epsilon)\cdot n^{\cO(1)}$ using a mix of ideas from Theorem \ref{thm:general} and Lemma \ref{lem:coreset}. Then we show how to use a simple greedy argument in order to obtain a $(1 + \epsilon)$-approximation in $f(k,\epsilon, d)\cdot  n^{\cO(1)}$ time. 

\subsection{$(2 + \epsilon)$-approximation}

The goal of this subsection is to prove the following Theorem. 

\twoapprox*

Let $B^\star_1, \dots B^\star_k$ denote the balls of an optimal solution and $\mu^\star$ the assignment associated to this solution. Note that we can assume that each $B^\star_i$ is the minimum enclosing ball of $(\mu^\star)^{-1}(B^\star_i)$ and that we know the radius $r_i$ of every $B^\star_i$.
As in Theorem \ref{thm:general}, we will distinguish between balls which get assigned a lot of points by $\mu^\star$ and the others. Let $\mathcal{B}^\star_1$ be the set of balls $B^\star_i$ such that $|(\mu^\star)^{-1}(B^\star_i)| \geq U/k$ and $\mathcal{B}^\star_2$ the rest. The next lemma shows the important properties of balls in $\mathcal{B}^\star_1$. 

\begin{lemma}\label{lem:sampling_99}
	For every integer $i \in [k]$ such that $B^\star_i \in \mathcal{B}^\star_1$, there exists a randomized algorithm running in time $O(nd/\epsilon^3)$ that outputs with probability $(\frac{1}{2k^2})^{32/\epsilon^2}$ a subset $S_i$ of $B^\star_i$ such that if $B = \MEB(S_i)$, then $\ext(B,\epsilon r_i)$ contains all but $U/2k$ points of $B^\star_i$.
\end{lemma}

\begin{proof}
	Remember that we can assume that $|P| \leq k\cdot U$. The algorithm starts by sampling uniformly at random a point $x$ in $P$. Since $B^\star_i \in \mathcal{B}^\star_1$, with probability at least $1/k^2$, $x \in B^\star_i$. Suppose from now on that this is the case. If the ball $B$ of center $x$ and radius $\epsilon r_i$ contains all but $U/2k$ points of $B^\star_i$, then $S_i :=\{x\}$ satisfies the property of the lemma. If not, it means that $U/2k$ points of $B^\star_i$ are outside this ball. Let $y$ be an element selected uniformly at random in $P \setminus B$. With probability at least $1/2k^2$, $y \in B^\star_i$. Suppose that it is the case and let $S_i:= S_i \cup \{y\}$. Note that, $\MEB(S_i)$ has radius at least $\frac{\epsilon r_i}{2}$ since $x$ and $y$ are at a distance at least $ \epsilon r_i$. We can now repeat the process of sampling outside of $\ext(\MEB(S_i),\epsilon r_i)$ to obtain another point of $B^\star_i$ with probability at least $\frac{1}{2k^2}$ until $\ext(\MEB(S_i),\epsilon r_i)$ contains all but $U/2k$ points of $B^\star_i$. 
	
	Moreover, since after sampling $y$, $\MEB(S_i)$ has radius at least $\frac{\epsilon r_i}{2}$, Lemma \ref{lem:corest_next} implies that $\MEB(S_i)$ increases its radius by at least $\epsilon^2r_i/32$ at each step. Therefore, this procedure must end in at most $32/\epsilon^2$ steps. The probability that this algorithm succeeds is equal to the probability that the sampling succeeds at each step and that the algorithm decides correctly if $S$ is the desired output. Overall, this gives a probability of success of at least  $(\frac{1}{2k^2})^{32/\epsilon^2}$. Moreover, the algorithm needs to compute at each step an approximation of the minimum enclosing ball $S_i$, which can be done in $O(nd/\epsilon) $ time  using the result in \cite{DBLP:journals/jea/KumarMY03} and removing the points inside this balls from the sampling pool. Overall this gives a running time of $O(nd/\epsilon^3)$.
	 
\end{proof}

Applying the previous result to every ball in $\mathcal{B}^\star_1$ allows us to obtain the following result.

\begin{lemma}\label{lem:cover}
	There exists an algorithm, running in time $O(\frac{knd}{\epsilon^3})$ that outputs, with probability at least $(\frac{1}{4k^3})^{32k/\epsilon^2}$:
	\begin{itemize}
		\item A list $I$ of indices in $[k]$ s.t, if $ B^\star_i \in \mathcal{B}^\star_1$, then $i \in I$
		\item For every $i \in I$, a list $S_i \subseteq B^\star_i$ such that if $B^\star_i \in \mathcal{B}^\star_1$ , then $\ext(\MEB(S_i), \epsilon r_i)$ contains all the points of $B^\star_i$ but $U/2k^2$.
		\item The union of $\ext(\MEB(S_i), \epsilon r_i)$ for $i \in I$ contains all the points in $P$
	\end{itemize}
\end{lemma}

\begin{proof}
	The algorithm starts by applying Lemma \ref{lem:sampling_99} to every ball in $\mathcal{B}^\star_1$ and set $I$ to be the set of indices $i$ such that $\mathcal{B}^\star_1$ and $S_i$ the set obtained with the algorithm. This means that, with probability at least $(\frac{1}{2k^2})^{32k/\epsilon^2}$ and with a running time of $O(\frac{knd}{\epsilon^3})$, $I$ and the $S_i$ satisfy the first two properties of the lemma. 
	
	In order to satisfy the third condition, we will do the following greedy procedure. As long as there exists a point $x$ which is not covered by the union of the $\ext(\MEB(S_i), \epsilon r_i)$ for $i \in I$, the algorithm will guess the index $i$ such that $\mu^\star(x) = B^\star_i$. If $i \in I$ then we add $x$ to $S_i$ and if not we add $i$ to $I$ and set $S_i = \{x\}$. Note that at this point, $I$ can contain indices corresponding to balls outside of $\mathcal{B}^\star_1$. 
	
	As in the proof of \ref{lem:sampling_99}, because we only add points to $S_i$ which are outside $\ext(\MEB(S_i), \epsilon r_i)$ this process cannot add more than $32/\epsilon^2$ points to $S_i$, which means that this process stops after $32k/\epsilon^2$ steps. At that time the third property of the lemma is satisfied. The algorithm therefore succeeds if at each step the correct index $i\in [k]$ is selected. The success probability of the latter event is at least $(\frac{1}{k})^{32 k/\epsilon^2}$, as each point can be assigned to one of the $k$ clusters.  Overall, the probability of success is at least $(\frac{1}{2k^2})^{32k/\epsilon^2} \cdot (\frac{1}{k})^{32 k/\epsilon^2}\ge (\frac{1}{4k^3})^{ 32 k/\epsilon^2}$. Moreover, the algorithm needs to compute at each step an approximation of the minimum enclosing balls and see if the union contains $P$. This can be done in $O(\frac{knd}{\epsilon^3})$.
\end{proof}

Suppose that the algorithm from Lemma \ref{lem:cover} succeeds and we obtain the set $I$ as well as a set $S_i$ for every $i \in I$. One important remark is that the set of points $x$ such that $\mu^\star(x) = B^\star_i$ and $i \not \in I$ or $x \not \in \ext(\MEB(S_i), \epsilon r_i)$ has size at most $U$ as for every $i$ such that $B^\star_i \in \mathcal{B}^\star_1$, $\ext(\MEB(S_i), \epsilon r_i)$ contains all the points except at most $U/2k$.

Let us now define a graph $G'$ where the vertices of $G'$ correspond to indices in $[k]$ and the edges are defined as follows:
\begin{itemize}
	\item If $i,j \in I$, then $i$ and $j$ are adjacent if $\ext(\MEB(S_i), \epsilon r_i) \cap \ext(\MEB(S_j), \epsilon r_j)$
	\item If $i \in I$ and $j \not \in I$, then $i$ and $j$ are adjacent if $\ext(\MEB(S_i), \epsilon r_i) \cap B^\star_j$
	\item If $i,j \not \in I$, then $i$ and $j$ are adjacent if $B^\star_i \cap B^\star_j$
\end{itemize}

Let $C_1, \dots, C_s$ denote the set of connected components. By abusing notations, we will also use $C_j$ to denote the set of points in $P$ contained in the union of $\ext(\MEB(S_i), \epsilon r_i)$ for all the indices $i \in I \cap C_j$. For every connected component $C_j$, we note by $\rad(C_j)$ the sum of $\rad(B^\star_i)$ over all indices $i \in C_j$. Let us first show the following lemma about covering two intersecting balls:

\begin{lemma}\label{lem:intersecting_balls}
Let $B_1$ be a ball of radius $r_1$ centered at $c_1$ and $B_2$ a ball of radius $r_2$ centered at $c_2$ such that $B_1$ and $B_2$ intersects. There exists a ball of radius $r_1 + r_2$ containing $B_1$ and $B_2$.
\end{lemma}

\begin{proof}
Indeed, because $B_1$ and $B_2$ intersect, it means that the distance between $c_1$ and $c_2$ is smaller than $r_1 + r_2$ and let $c$ be a point at distance $r_2$ from $c_1$ and $r_1$ from $c_2$. By triangle inequality the ball of radius $r_1 + r_2$ centered at $c$ contains $B_1$ and $B_2$.
\end{proof}

The previous lemma allows us to show:

\begin{lemma}\label{lem:covering_component}
	For any connected component $C_j$, there is a ball of radius $(1+ \epsilon) \rad(C_j)$ covering all the points in $C_j$. 
\end{lemma}

\begin{proof}
	This can be proved by induction on the number of indices in $C_j$ and using Lemma \ref{lem:intersecting_balls} as well as the fact that $C_j$ is connected. 
\end{proof}

Let $\mathcal{C}_1$ denote the list of connected components $C_j$ containing at least one element $i\notin I$ and $\mathcal{C}_2$ the other components. 
Let $\mathcal{L}$ be the union of:

\begin{itemize}
	\item The list $L_1$ containing for every $i \in I$ such that the component $C_j$ containing $i$ is in $\mathcal{C}_1$, the ball $B_i = \ext(\MEB(S_i), \epsilon r_i)$
	\item The list $L_2$ containing for every component $C_j \in \mathcal{C}_1$, the ball $B_{C_j}$ containing all the points of the component $C_j$.
	\item The list $L_3$ containing for every $i \in I$ such that the component $C_j$ of $i$ is in $\mathcal{C}_2$, a ball $B_i$ centered at any arbitrary point of $S_i$ and of radius $2 r_i$.
\end{itemize}

\begin{lemma}\label{lem:L_valid}
$\mathcal{L}$ is a valid solution to the \radii problem.
\end{lemma}
\begin{proof}
	First, note that $L_1\cup L_3$ contains $|I|$ balls, and $L_2$ contains at most $k-|I|$ balls. Thus, $\mathcal{L}$ contains at most $k$ balls. 
	To complete the proof of the lemma, let us define a valid assignment $\mu$ to the balls in $\mathcal{L}$. First, for every $x$ such that $\mu^\star(x) = B^\star_i$ with $i \in I$ and the component of $i$ is in $\mathcal{C}_2$, then $B^\star_i \subseteq B_i$ and we let $\mu(x) = B_i$. For every $x$ such that $\mu^\star(x) = B^\star_i$ with $i \in I$ and the component of $I$ is in $\mathcal{C}_1$, we let $\mu(x) = B_i$ if $x \in B_i$. Remember that in that case, $B_i = \ext(\MEB(S_i), \epsilon r_i)$, which contains all the points of $B^\star_i$ except at most $U/2k$. At that point, all the $\mu^{-1}(B_i)$ are subsets of $(\mu^\star)^{-1}(B^\star_i)$, which means in particular that no ball has been assigned more than $U$ points. Moreover, the set of points which have not been assigned yet either belongs to $(\mu^\star)^{-1}(B^\star_i)$ for some $i \notin I$ or some $B^\star_i \setminus  \ext(\MEB(S_i), \epsilon r_i)$ where $i\in I$ and is in a component $C_j \in \mathcal{C}_1$. This implies first that there is less than $U$ points for which $\mu(x)$ has not been defined at that point, and for each such point $x$, it belongs to a component $C_j$ that intersects a ball $B^\star_i$ with $i\notin I$. In particular, there is a ball $B_{C_j}$ in $L_2$ containing $x$. For each of these point $x$, we let $\mu(x) = B_{C_j}$. By the previous discussion, $\mu$ is a valid assignment. 
\end{proof}

\begin{lemma}\label{lem:L_value}
	$\rad(\mathcal{L}) \leq  (2+2\epsilon)\rad(\mathcal{B}^\star)$
\end{lemma}

\begin{proof}
	The proof is implied by the following inequalities: 
	\begin{itemize}
		\item $\rad(L_3) \leq 2\cdot \sum_{C_j \in \mathcal{C}_2} \rad(C_j)$
		\item $\rad(L_2) \leq (1 + \epsilon )\sum_{C_j \in \mathcal{C}_1} \rad(C_j) $
		\item $\rad(L_1) \leq (1 + \epsilon )\sum_{C_j \in \mathcal{C}_1} \rad(C_j)  $
	\end{itemize}
The first inequality follows from the fact that $L_3$ contains a ball of radius $2r_i$ for every element of $C_j$, the second one follow from Lemma \ref{lem:covering_component} and the last one follows from the fact that $B_i = \ext(\MEB(S_i), \epsilon r_i)$ for all elements of $L_1$. 

This ends the proof, as $(\mathcal{C}_1,\mathcal{C}_2)$ is a partition of the  elements of $\mathcal{B}^\star$.
\end{proof}

\begin{lemma}
	 $\mathcal{L}$ can be computed in time $2^{\cO((k/\epsilon^2)\log (k/\epsilon))}\cdot dn^3$ with constant probability. 
\end{lemma}

\begin{proof}
	Again, this follows from the previous discussions. The algorithm starts by using Lemma \ref{lem:guess_Rd} and get, with probability at least $\frac{\epsilon^k}{(2k)^k\cdot n^2}$ an approximate radius $r_i$ for every ball in linear time. Then the algorithm guess uniformly at random the bi-partition of $\mathcal{B}^\star$ into $(\mathcal{B}^\star_1, \mathcal{B}^\star_2)$, and succeeds with probability at least $1/2^k$. Knowing this partition, the algorithm can apply Lemma \ref{lem:cover} 
	to obtain in $O(\frac{knd}{\epsilon^3})$ time and with probability at least $(\frac{1}{4k^3})^{32 k/\epsilon^2}$ the lists $I$ and $S_i$ satisfying the property of the lemma. The algorithm then guesses the components of $G'$, which can be done with probability at least $1/k^k$. Once this is done, the algorithm decides of the bi-partition of the connected components into $(\mathcal{C}_1, \mathcal{C}_2)$, which again can be done with probability at least $1/2^k$. Knowing this partition as well as the sets $S_i$, the algorithm can then construct the lists $L_1, L_2$ and  $L_3$ deterministically. Overall, by running the previous algorithm $\frac{n^2k^{\cO(k/\epsilon^2)}}{\epsilon^{\cO(k)}}=2^{\cO((k/\epsilon^2)\log (k/\epsilon))}\cdot n^2$ times, we can find the set $\mathcal{L}$ with constant probability, which ends the proof.
\end{proof}

The above lemma completes the proof of Theorem \ref{thm:twoapprox}. 

\subsection{The Bi-criteria Approximation}

The ideas developed to prove Theorem \ref{thm:twoapprox} can easily be adapted to prove a $(1 + \epsilon, 1 + \epsilon)$ bi-criteria approximation in Euclidean space. Let us modify the definition of  $\mathcal{B}^\star_1$ as follows: Let $\mathcal{B}^\star_1$ denote the set of balls $B^\star_i$ of $\mathcal{B}^\star$ such that $(\mu^\star)^{-1}(B^\star_i) \geq \epsilon U/k$ and $\mathcal{B}^\star_2 =  \mathcal{B}^\star \setminus \mathcal{B}^\star_1$. Then we can trivially adapt Lemma \ref{lem:sampling_99} to obtain the following lemma, as the only change is the probability to sample inside the desired sets. 

\begin{lemma}
	For every integer $i \in [k]$ such that $B^\star_i \in \mathcal{B}^\star_1$, there exists a randomized algorithm running in time $O(\frac{nd}{\epsilon^3})$ that outputs with probability $(\frac{\epsilon}{2k^2})^{32/\epsilon^2}$ a subset $S_i$ of $B^\star_i$ such that if $B = \MEB(S_i)$, then $\ext(B,\epsilon r_i)$ contains all but $\frac{\epsilon U}{2k}$ points of $B^\star_i$.
\end{lemma}

And again, applying the previous lemma to every ball in $\mathcal{B}^\star_1$ allows us to obtain the following result similarly to Lemma \ref{lem:cover}, again assuming the algorithm knows the radii of the balls in an optimal solution.

\begin{lemma}
	There exists an algorithm, running in time $O(\frac{knd}{\epsilon^3})$ that outputs, with probability at least $(\frac{\epsilon}{4k^3})^{k \cdot \epsilon^2/32}$:
	\begin{itemize}
		\item A list $I$ of indices in $[k]$ s.t, if $ B^\star_i \in \mathcal{B}^\star_1$, then $i \in I$
		\item For every $i \in I$, a list $S_i \subseteq B^\star_i$ such that if $B^\star_i \in \mathcal{B}^\star_1$ , then $\ext(\MEB(S_i), \epsilon r_i)$ contains all the points of $B^\star_i$ but $\frac{\epsilon U}{2k}$.
		\item The union of $\ext(\MEB(S_i), \epsilon r_i)$ for $i \in I$ contains all the points in $P$
	\end{itemize}
\end{lemma}

Suppose the algorithm of the previous lemma succeeds and let $ \mathcal{B}'$ denote the set of all $\ext(\MEB(S_i), \epsilon r_i)$ for all $i \in S_i$. Let $\mu'$ be the assignment from  $P$ to $ \mathcal{B}'$ defined as follows:

\begin{itemize}
    \item Whenever $x \in P$ is such that $x \in \ext(\MEB(S_i), \epsilon r_i)$ and $\mu^*(x) = B^*_i$, then $\mu'(x) = (\MEB(S_i), \epsilon r_i)$
    \item For any other point, $\mu'(x) = \ext(\MEB(S_j), \epsilon r_j)$ for some $j \in I$ such that $x \in \ext(\MEB(S_j), \epsilon r_j) $
\end{itemize}

Since $S_i \in B^*_i$ for all $i \in I$, we have that $\MEB(S_i) \subseteq B^*_i$ and thus the sum of radii in $ \mathcal{B}'$ is at most $(1 + \epsilon) OPT$. Moreover, by construction, the set $X$ of points $x$ such that $x \in B^*_i$ and $x \not \in \ext(\MEB(S_i), \epsilon r_i)$ for $i \in [k]$ is at most $\epsilon U$, and thus no balls of $ \mathcal{B}'$ is assigned more than $(1 + \epsilon) U$ points. This ends the description of our bicriteria algorithm.

\genbi*
\subsection{Approximation Scheme with Dependence on the Dimension}

In this subsection, we show the following result:

\ptaskandd*
\begin{proof}

First, remember that by applying Lemma \ref{lem:guess_Rd}, we can assume that we know the radii $r_1 \geq r_2, \dots, \geq r_k$ of the solution and that every radius is a multiple of $\frac{r_1 \epsilon}{k}$. In particular it means that no $r_i$ is smaller than $\frac{r_1 \epsilon}{k}$. By a standard greedy procedure we can find a set $T_1, \dots, T_k$ of balls of radius $2r_1$ covering all the points in $P$, or conclude that the problem doesn't have a solution. Then the idea is to guess for each center $c^\star_i$, its location up to a precision of $\frac{r_1 \epsilon}{k}$. If we do so correctly, and take a ball of radius $r_i + \frac{r_1 \epsilon}{k}$ for every ball, we obtain a set of balls $B_1, \dots, B_k$ such that every $B^\star_i$ is contained in $B_i$ and the sum of radii is at most $\sum_{i \in [k]} r_i + k \cdot  \frac{r_1 \epsilon}{k} \leq \sum_{i \in [k]} r_i + \epsilon r_1$.

In order to restrict the space of possible centers, let graph $G$ be the graph where every vertex $u_i$ corresponds to a ball $T_i$ for $i \in [k]$ and $u_iu_j$ is an edge of $G$ for every $i,j \in [k]$ if the distance between $T_i$ and $T_j$ is smaller than $2r_1$. Let $C_1, C_t$ denote the set of connected components of $G$. Note that taking for each component $C_j$, the set of points $P_j$ inside the balls associated to $C_j$ induces a partition of the vertices of $P$. Indeed the set of balls containing the same point $p \in P$ intersect and thus belong to the same component. 

\begin{claim}
	For every component $C_j$, there exists a ball $R_j$ of radius at most $(4k-1)r_i$ covering all the points in $P_j$
\end{claim}

\begin{proof}
	Consider a spanning tree of $C_j$ and for every edge $u_iu_j$ of the spanning tree, add a new ball $T_{i,j}$ of radius $r_i$ intersecting both $T_i$ and $T_j$. Now if we take the intersection graph of the $T_i$ and $T_{i,j}$, we have a set of at most $2k$ balls of radius $r_1$ and we can apply induction with Lemma \ref{lem:covering_component} to end the proof.
\end{proof}

By definition of $G$, no ball of the optimal solution can cover points in two different $P_j$. Moreover, any ball of radius at most $r_1$ covering only points in $P_j$ must have its center in $\ext(R_j, r_1)$. Since $\ext(R_j, r_1)$ is a ball of radius at most $4kr_i$, we can cover it with at most $(\frac{k}{c \epsilon})^d$ balls of radius $\frac{\epsilon r_1}{k}$ for some constant $c$. Then for every center $c^\star_i$ covering points in $P_j$, the algorithm will simply guess which of these balls it belongs to. As there is at most $(\frac{k}{c \epsilon})^d$ balls for each $C_j$, we have $(k\cdot (\frac{k}{c \epsilon})^d)^k$ possible choices for the $c_i$. Overall this gives an algorithm in time $(\frac{k}{\epsilon})^{\cO(kd)}\cdot n^3$.
\end{proof}
\section{Hardness Results}
\label{sec:hardness}

In this section, we describe our hardness results. These results show that the approximation bounds we have achieved are tight modulo small constant factors (even if we are allowed to use $f(k)n^{O(1)}$ time for any computable function $f$). Moreover, in general metric, we could show a bound based on Exponential Time Hypothesis (ETH) \cite{impagliazzo2001complexity}, that eradicates the possibility of solving the problem exactly in $f(k)n^{o(k)}$ time. First, we describe this result.     

\subsection{Hardness in General Metrics}

Our reduction is from the \emph{Dominating set} problem in general graphs. In Dominating set, given an $n$-vertex graph $G=(V,E)$ and a parameter $k$, the goal is to decide whether there is a subset $V'\subseteq V$ of size $k$ such that for any vertex $v\in V$, either $v \in V'$ or $v$ is adjacent to a vertex of $V'$. If such a subset $V'$ exists, it is called a dominating set. 

Given an instance of Dominating set, we construct a new graph $G'$ by adding a set of $kn^2$ auxilliary vertices $V_a$. Moreover, each vertex in $V$ is connected to all the vertices in $V_a$ in $G'$. Then, we construct an instance of \radii by setting $P$ to be $V\cup V_a$, $d$ to be the shortest path metric in $G'$, and $U$ to be $n^2+n$. Also the number of balls $k$ is set to be the solution size parameter in Dominating set. In the following, we use the terms vertex and point interchangeably, as the points in $P$ are exactly the vertices in $G'$. 

\begin{lemma}
 There is a dominating set in $G$ of size $k$ if and only if there is a solution to the \radii instance of cost $k$. 
\end{lemma}

\begin{proof}
 Suppose there is a dominating set $V'$ of size $k$ in $G$. Consider the set of $k$ balls $\mathcal{B}=\{B(v',1)\mid v'\in V'\}$. The cost of $\mathcal{B}$ is $k$. We prove that $\mathcal{B}$ is a valid solution to \radii by showing the existence of an assignment $\mu$ that satisfies the capacity constraints. For each vertex $v\in V$, assign it to a ball $B(v',1)$ such that either $v=v'$ or $v$ is adjacent to $v'$. Note that each ball is assigned at most $n$ vertices so far. Then consider any partition of $V_a$ into $k$ equal sized subsets and assign the vertices in each subset to a unique ball in $\mathcal{B}$. Thus, $\mu$ assigns each vertex in $P$ to a ball, and each ball is assigned at most $n^2+n$ vertices. As all vertices of $V_a$ are adjacent to each vertex $v'\in V'$ by construction, the vertices in $V_a$ are in the ball $B(v',1)$.  Hence, $\mu$ is a valid assignment.      
 
 Now, suppose there is a solution to \radii of cost $k$. First, we claim that the radius of each ball in this solution is at least 1, and hence exactly 1 due to the cost $k$ of the solution. Suppose not, i.e., there is a ball of radius 0 in the solution. But, then no other vertex than the center of this ball can be assigned to it. So, the vertices in $V_a$ are assigned to the remaining $k-1$ balls. By pigeon hole principle, there is a ball that receives at least $kn^2/(k-1)>=n^2+(n^2/(k-1))>n^2+n$ vertices. But, this is a contradiction, as the capacity is $n^2+n$. Thus, the solution contains exactly $k$ balls of radius 1. For a similar reason, the center of these balls must be in $P\setminus V_a$. As each vertex in $V$ is assigned to a ball of radius 1 in the solution, either it is a center of a ball or is adjacent to a center. It follows that the set of centers of these balls is a dominating set of size $k$. 
\end{proof}

The above lemma shows that \radii is as hard as Dominating set. Now, it is widely known that it is not possible to solve Dominating set in $f(k)n^{o(k)}$ time, unless ETH is false \cite{cygan2015parameterized}. Thus, we obtain the same lower bound for \radii. Moreover, the constructed instance in the above reduction has the property that all input distances are bounded by a polynomial in the input size. Thus, \radii is also strongly   \np-hard. So, it does not admit any FPTAS, unless \p=\np. Note that an FPTAS is a $(1+\epsilon)$-approximation in $n^{O(1)}$ time for any $\epsilon > 0$.  

\genhard*

\subsection{Hardness in the Euclidean Case}

We show by a reduction from $k$-clique that \radii does not admit any FPTAS even when the number of balls in the solution is only 2. The reduction is partly motivated by the one in \cite{shenmaier2013problem} that proves hardness of Smallest $k$-Enclosing ball, where given a set of $n$ points in $\R^d$, the goal is to find a ball of smallest radius that encloses at least $k < n$ points. Our reduction is similar, except we need to take one more ball to cover all the points. However, we cannot just use the hardness of Smallest $k$-Enclosing ball directly, as we need some additional properties of the given instance which we have to reduce. Instead, we follow their direction and show a separate reduction from $k$-clique. Our details are much more involved. 

In $k$-clique, we are given an $n$-vertex graph $G=(V,E)$, and the goal is to decide whether there is a clique of size $k$. Such a clique is called a $k$-clique. We can assume wlog that $k \le n/2$. If $k > n/2$, we can make two copies of $G$ and join two vertices in these two copies by a path of length 3, and solve the $k$-clique problem in the new graph. It is not hard to see that $G$ has a $k$-clique if and only if the new graph has a $k$-clique. The idea is, no clique in the new graph can contain the edges of the 3-path. Thus, any $k$-clique must be contained in a single copy of the graph. Now, for the new graph $k$ is at most half of the total number of vertices. 

Brandes et al.~\cite{brandes2016cliques} showed a regularization lemma, which given any graph $G$, constructs an $s$-regular graph $G_s$ such that $G$ has a $k$-clique if and only if $G_s$ has a $k$-clique. By applying, this regularization lemma, we can also assume that the graph is $s$ regular. It follows that this restricted version is also \np-hard. 

Our reduction is as follows. Let $|V|=n$ and $|E|=m$. For each vertex $v\in V$, we construct a binary point of dimension $m+2$ where for $1\le i\le m$ the value of the $i$-th coordinate is 1 if the $i$-th edge contains $v$ as an endpoint; otherwise the value is 0. Note that the values of $(m+1)$-th and $(m+2)$-th coordinates are 0 for all these points. Let $P'$ be the set of these points. Let $\widehat{\Delta}$ be the maximum interpoint distance (diameter) of these points, and ${\Delta}=\widehat{\Delta}+2\sqrt{n}$. We also add four special points whose last 2 coordinates are $(\Delta,\Delta), (\Delta,-\Delta), (-\Delta,\Delta), (-\Delta,-\Delta)$ and other coordinates are 0. Let us denote these points by $p_{++}, p_{+-}, p_{-+}$ and $p_{--}$, respectively. Denote the set by $P$ obtained by adding these four points to $P'$. Set the capacity $U$ to $n+4-k$. Also the number of balls that need to be selected in the constructed instance is 2. Let $c_0$ be the point whose all coordinates are 0.   

\begin{observation}\label{obs:1cover}
 One needs a ball of radius at least $\sqrt{2}\Delta$ to enclose the points in $P$. Also the ball $B(c_0,\sqrt{2}\Delta)$ encloses all the points in $P$. 
\end{observation}

\begin{proof}
Note that the distance between $p_{++}$ and $p_{--}$ is $\sqrt{(2\Delta)^2+(2\Delta)^2}=2\sqrt{2}\Delta$. Thus the diameter of this point set is at least $2\sqrt{2}\Delta$, and hence one needs a ball of radius at least $\sqrt{2}\Delta$ to enclose all the points. 

It is easy to see that $B(c_0,\sqrt{2}\Delta)$ contains the four special points, as their distances from $c_0$ are exactly $\sqrt{2}\Delta$. Also, the distance between any point $p\in P'$ and $c_0$ is at most $\sqrt{n-1}\le \sqrt{2}\Delta$, as at most $s\le n-1$ coordinates of $p$ can be one. Hence, the observation follows. 
\end{proof}

We need the following lemma from \cite{shenmaier2013problem}. 

\begin{lemma}\label{lem:minballradius}
 \cite{shenmaier2013problem} Denote by $R$ the radius of a smallest radius ball enclosing $k$ points of $P'$ and let $A=(1-1/k)(s-1)$. Then $R^2\le A$ if $G$ contains a $k$ clique; otherwise, $R^2\ge A+1/k^2$.   
\end{lemma}

The next lemma shows the connections between the two instances. 

\begin{lemma}
If there is a $k$-clique in $G$, then there is a solution to the \radii instance of cost at most $\sqrt{2}\Delta+\sqrt{A}$. Otherwise, any solution has cost at least $\sqrt{2}\Delta+\sqrt{A+1/k^2}$.   
\end{lemma}

\begin{proof}
 Suppose there is a $k$-clique in $G$. By Lemma \ref{lem:minballradius}, there is a ball $B$ of radius at most $\sqrt{A}$ that contains $k$ points. We select the ball $B(c_0,\sqrt{2}\Delta)$ and $B$ as the solution. The cost is at most $\sqrt{2}\Delta+\sqrt{A}$. Assign any $k$ points in $B$ to it.  Assign the remaining $n+4-k=U$ points to $B(c_0,\sqrt{2}\Delta)$.  By Observation \ref{obs:1cover}, $B(c_0,\sqrt{2}\Delta)$ contains all these $n+4-k$ points. Note that $n+4-k\ge n/2$, as $k\le n/2$, and thus $B$ is assigned $k\le n/2\le n+4-k$ points. Thus the capacities are satisfied, and hence the solution constructed is a valid solution. 
 
 Now, suppose there is no $k$-clique. Consider any solution $S$ with two balls $B_1$ and $B_2$. We show that the cost of $S$ is at least $\sqrt{2}\Delta+\sqrt{A+1/k^2}$. First, we prove that one can assume that $B(c_0,\sqrt{2}\Delta)$ is selected in $S$. Wlog, let $p_{++}$ is in $B_1$. If $p_{--}$ is also in $B_1$, then radius of $B_1$ is at least $\sqrt{2}\Delta$, and we can replace $B_1$ by $B(c_0,\sqrt{2}\Delta)$. So, assume that $p_{--}$ is in $B_2$. Similarly, one can argue that $p_{+-}$  and $p_{-+}$ are in different balls in $S$. Suppose, $p_{+-}$ is in $B_1$ and $p_{-+}$ is in $B_2$. Then the radius of $B_1$ is at least $\sqrt{(\Delta+\Delta)^2}/2=\Delta$, and the radius of $B_2$ is at least $\Delta$. So, cost of $S_1$ is at least $2\Delta$. Now,
 
 \begin{align*}
  \sqrt{A+1/k^2} < \sqrt{s+1/k^2}
  < \sqrt{s+1}\le \sqrt{s}+1< \sqrt{
  n}\le \Delta/2.
 \end{align*}
 
The last inequality follows, as $\Delta \ge 2\sqrt{n}$.  
Thus, in this case the cost of $S$ is at least $2\Delta > \sqrt{2}\Delta+\Delta/2>\sqrt{2}\Delta+\sqrt{A+1/k^2}$. Similarly, the cost is at least $\sqrt{2}\Delta+\sqrt{A+1/k^2}$ if $p_{+-}$ is in $B_2$ and $p_{-+}$ is in $B_1$. Hence, wlog, we can assume that $B_1=B(c_0,\sqrt{2}\Delta)$ is in $S$, and $B_1$ contains both $p_{++}$ and $p_{--}$. 

Now, $B(c_0,\sqrt{2}\Delta)$ can be assigned at most $U=n+4-k$ points. Thus, $B_2$ must be assigned at least $k$ points. Suppose $p_{+-}$ is in $B_2$. Note that $p_{-+}$ is in $B_1$, otherwise, the radius of $B_2$ is at least $\sqrt{2}\Delta\ge \sqrt{A+1/k^2}$ and we are done. So, assume that  $p_{+-}$ is in $B_2$. In this case, as $B_2$ contains at least one more point (of $P'$), the diameter of $B_2$ is at least $\sqrt{\Delta^2+\Delta^2}=\sqrt{2}\Delta$. Thus, the radius of $B_2$ is at least $\Delta/\sqrt{2}> \Delta/2>\sqrt{A+1/k^2}$, and we are done. Similarly, we can rule out the case when $p_{-+}$ is in $B_2$. Hence, the only case we are left with is where $B_2$ contains  points  only from $P'$. But, in this case, by Lemma \ref{lem:minballradius}, the radius of $B_2$ is at least $\sqrt{A+1/k^2}$. Hence, the lemma follows. 
\end{proof}

This already proves the \np-hardness of our problem in the Euclidean case. Next, we argue that our problem does not admit any FPTAS. To do that it is sufficient to show a multiplicative gap between the cost of the \radii instances corresponding to Yes instances of $k$-clique and the cost of the \radii instances corresponding to No instances of $k$-clique. 

First, note that the diameter $\widehat{\Delta}$ of the points in $P'$ is bounded by $n^2$, as each point is a 0-1 vector with at most $n-1$ 1 components. Thus, $\Delta\le \alpha n^2$ for some constant $\alpha > 1$. . Also, $\sqrt{A}\le n$. Thus, 

\begin{align*}
 \sqrt{2}\Delta+\sqrt{A+1/k^2} & >  \sqrt{2}\Delta+\sqrt{A}(1+1/(3Ak^2))\\ 
 & = \sqrt{2}\Delta+\sqrt{A}+\sqrt{A}/(6Ak^2)+1/(6\sqrt{A}k^2)\\
 & > \sqrt{2}\Delta+\sqrt{A}+\sqrt{A}/(6n^3)+\sqrt{2}\Delta/(6n^3\cdot \sqrt{2}\Delta)\\
 & > \sqrt{2}\Delta+\sqrt{A}+\sqrt{A}/(12\alpha  n^5)+\sqrt{2}\Delta/(12\alpha n^5)\\
 & > (\sqrt{2}\Delta+\sqrt{A})(1+1/(12\alpha n^5))
\end{align*}

Thus, we could show a gap of $(1+1/(12\alpha  n^5))$. Thus, if there is a $(1+\epsilon)$-approximation algorithm for \radii for all $\epsilon > 0$ in $n^{O(1)}$ time, we can distinguish between Yes and No instances of $k$-clique in polynomial time. Hence, we have the following theorem. 

\euclidhard*
\section{Conclusions}
\label{sec:conclude}
In this paper, considering the \radii problem, we obtained the first constant-factor ($15+\epsilon$) approximation algorithm that runs in \FPT time, making significant progress towards understanding the barriers of capacitated clustering. While our techniques are tailor-made for FPT type results, we hope some of the ideas will also be useful in obtaining a similar approximation in polynomial time. We leave this as an open question. 

\begin{tcolorbox}
	\begin{description}
	\setlength{\itemsep}{-2pt}
	\item[Question $1$:] Does \radii admit 
	a polynomial time constant-approximation algorithm, even with uniform capacities?
	\end{description}
\end{tcolorbox}

For the problem with uniform capacities, we obtained improved approximation bounds of $4+\epsilon$ and $2+\epsilon$ in general and Euclidean metric spaces, respectively. We also showed hardness bounds in both general and Euclidean metric spaces complementing our approximation results. The following two natural open questions are left by our work.

\begin{tcolorbox}
\begin{description}
\setlength{\itemsep}{-2pt}
\item[Question $2$:] What is the best constant-factor approximation possible for \radii or uniform \radii 

in \FPT\ time? 
\item[Question $3$:] Does Euclidean \radii admit 
an $(1+\epsilon)$-approximation algorithm, in $f(k,\epsilon)\cdot n^{g(\epsilon)}$ time for some functions $f$ and $g$? 
\end{description}
\end{tcolorbox}

\bibliographystyle{plain}
\bibliography{balls}

\end{document}